\documentclass[aps,pre,reprint,superscriptaddress]{revtex4-1}
\usepackage{amsmath}	
\usepackage{amssymb}	
\DeclareMathOperator{\PV}{P.V.}
\usepackage{graphicx}	
\graphicspath{ {figs/} }
\usepackage{tabularx}	
\usepackage{bm}	
\usepackage{amsthm}	
\newtheorem*{thm*}{Theorem}
\usepackage[super]{nth}	
\usepackage{enumitem}	
\usepackage[normalem]{ulem}	

\begin{document}

\title{Unwinding the model manifold: choosing similarity measures to remove local minima in sloppy dynamical systems}

\author{Benjamin L. Francis}
\author{Mark K. Transtrum}
\email{mktranstrum@byu.edu}
\affiliation{Department of Physics and Astronomy, Brigham Young University, Provo, Utah 84602, USA}

\date{\today}

\begin{abstract}
  In this paper, we consider the problem of parameter sensitivity in models of complex dynamical systems through the lens of information geometry.
  We calculate the sensitivity of model behavior to variations in parameters.
  In most cases, models are sloppy, that is, exhibit an exponential hierarchy of parameter sensitivities.
  We propose a parameter classification scheme based on how the sensitivities scale at long observation times.
  We show that for oscillatory models, either with a limit cycle or a strange attractor, sensitivities can become arbitrarily large, which implies a high effective-dimensionality on the model manifold.
  Sloppy models with a single fixed point have model manifolds with low effective-dimensionality, previously described as a ``hyper-ribbon''.
  In contrast, models with high effective dimensionality translate into multimodal fitting problems.
  We define a measure of curvature on the model manifold which we call the \emph{winding frequency} that estimates the density of local minima in the model's parameter space.
  We then show how alternative choices of fitting metrics can ``unwind'' the model manifold and give low winding frequencies.
  This prescription translates the model manifold from one of high effective-dimensionality into the ``hyper-ribbon'' structures observed elsewhere.
  This translation opens the door for applications of sloppy model analysis and model reduction methods developed for models with low effective-dimensionality.
\end{abstract}

\maketitle

\section{Introduction}
\label{sec:intro}
An essential part of the modeling process is selecting a similarity metric that quantifies the extent to which a model mimics the system or phenomenon of interest \footnote{Although our interest here is in mathematical models, the same is true of model organisms in biology that, by virtue of being (for example) mammals, vertebrates, eukaryotes, etc., have a similar phylogeny to a target system, e.g., humans.}.
The choice of similarity metric informs nearly all aspects of the modeling process: model selection, data fitting, model reduction, experimental design, model validation, etc..
Here, we consider the question of similarity metrics for dynamical systems, particularly oscillatory ones.
Although a common choice, the least squares metric comparing model outputs at selected times may lead to models with a \emph{high effective dimensionality}.
In addition to posing technical challenges (e.g., ill-posed, multimodal fitting problems), we argue that a high effective dimensionality reflects a more fundamental issue: that the choice of metric does not accurately capture the phenomenon of interest.
In this paper, we use sloppy model analysis and information geometry to identify parameter combinations in models of dynamical systems that lead to high effective dimensionalities.
We then use methods of signal processing to construct new similarity measures that ``unwind'' the model manifold and lead to well-posed inference problems.

Some have already observed that one's choice of metric is a critical aspect of parameter space exploration \cite{pospischil2008minimal, lemasson2001introduction}.
The relationship between model behavior and parameters is (locally) captured by sensitivity analysis.
Previous studies have decomposed the sensitivities of periodic signals into independent parts that control amplitude, period, and other features \cite{kramer1984sensitivity, larter1984sensitivity, wilkins2009sensitivity}.
In chaotic systems, it has been found that the dynamics exhibit an exponential sensitivity to parameters \cite{lea2000sensitivity, pisarenko2004statistical}.
In such cases, it is common to use measures of the statistical distribution in phase space, rather than time series \cite{lea2000sensitivity, lemasson2001introduction, annan2004efficient, lorenc20074d, lasota2013chaos}.
The present work combines these insights with tools of sloppy model analysis and information geometry.

Sloppy models are a broad class of models whose behavior exhibits an exponential hierarchy of parameter sensitivities \cite{brown2003statistical, brown2004statistical, frederiksen2004bayesian, waterfall2006sloppy, gutenkunst2007sloppiness, casey2007optimal, gutenkunst2007extracting, gutenkunst2007universally, daniels2008sloppiness}.
Using an information geometric approach, it has been shown that the local sensitivity analysis reflects a global property, i.e., a low effective dimensionality described as a \emph{hyper-ribbon} \cite{transtrum2010nonlinear, transtrum2011geometry, machta2013parameter, transtrum2015perspective}.
It has been suggested that this hyper-ribbon structure is why simple effective (i.e., low-dimensional) theories of collective behaviors exist for systems that are microscopically complicated \cite{machta2013parameter, transtrum2015perspective}.

The effective dimensionality of sloppy models has important statistical implications.
Information Criteria (such as Akaike or Bayes) are used in model selection and penalize those with too much fitting flexibility.
A model's fitting power is most easily estimated in the asymptotic limit, in which it is simply approximated by the number of parameters, i.e., the dimension of the model manifold.
For hyper-ribbons, these formula greatly overestimate fitting power of a model \cite{lamont2016information, lamont2017information}.
However, it is also possible for models to exhibit a \emph{high effective dimensionality}, i.e., have model manifolds whose fitting power is much larger than that suggested by the number of parameters.
As we show, these models will exhibit extreme multimodality when fit to data, and have model manifolds with large curvatures that tend to fill large volumes of behavior space.

The challenge of multimodality in fitting problems has been noted in many fields \cite{mendes1998non, moles2003parameter, annan2004efficient, lorenc20074d, ramsay2007parameter}.
Proposals for addressing multimodality have included global search methods \cite{yao1994nonlinear, he2007parameter, moles2003parameter, rodriguez2006hybrid, annan2004efficient}, increasing the size of the parameter space in order to escape local minima \cite{baake1992fitting, ramsay2007parameter}, and changing the parameter landscape through an alternative choice of metric \cite{lemasson2001introduction, annan2004efficient}.

In this study, we use model sensitivity analysis at long times to classify parameter combinations.
In turn, we classify models based on which parameter types they include.
We show that some classes of models exhibit an anomalous statistical dimension, that is, the effective dimensionality of the model may be either much more or less than the number of parameters, and argue for a deeper theoretical implication of this phenomenon.
Using an information geometric approach, we relate the effective statistical dimension to the curvature on the model manifold and explicitly demonstrate that alternative metrics can lead to different effective dimensions.
We present a prescription for how models of high effective dimension can be regularized through an appropriate choice of metric.

\section{Model and parameter classifications}
\label{sec:classes}
	\subsection{Similarity measures}
	\label{sec:cost}
	Consider a parametrized model of time $y(t; \theta)$ (which could be generated, for example, as the solution to a system of differential equations), where $\theta$ is a vector of parameters (which could include initial conditions) and $y$ is either a scalar or vector of observables.
We wish to quantify the similarity of the model behavior for different values of $\theta$.
The most common metric in the literature is least squares regression, in which case the distance (or cost) between two models, with parameters $\theta$ and $\theta_0$, takes the form
\begin{subequations}
  \label{eq:cost_int}
	\begin{align}
		\tag{\ref{eq:cost_int}}
		C(\theta) &= \frac{1}{2T} \int_0^T dt \bm{(}\delta y(t; \theta) \cdot \delta y(t; \theta)\bm{)}, \\
		\label{eq:cost_int_devs}
		\delta y(t; \theta) &\equiv y(t; \theta_0) - y(t; \theta).
	\end{align}
\end{subequations}
We are interested in the sensitivity of model predictions at different time scales.
By increasing the total time $T$, this cost function defines a coarse-graining in the effective sampling rate followed by a renormalization so that the total number of effective data points is constant.
When measuring the distance to observed data $y_i$ at times $t_i$ (with uncertainties $\sigma_i$ used as weights), the integral becomes a sum,
\begin{equation}
	\label{eq:cost_sum}
	C(\theta) = \frac{1}{2T}\sum_i{\left(\frac{y_i - y(t_i; \theta)}{\sigma_i}\right)^2},
\end{equation}
where we have employed the convention, e.g., $\delta y^2 \equiv \delta y \cdot \delta y$.

Being a distance measure, $C$ defines a metric on the space of data and model predictions known as \emph{data space} \cite{transtrum2010nonlinear, transtrum2011geometry}.
We interpret the model predictions $y(t_i; \theta)$ and observations $y_i$ as components of two vectors in data space which we denote $\mathbf{y}(\theta)$ and $\mathbf{y}$, respectively.
By varying $\theta$, $\mathbf{y}(\theta)$ sweeps out a surface in data space known as the \emph{model manifold}.
With this notation, Eq.~\eqref{eq:cost_sum} may be written as
\begin{subequations}
	\label{eq:cost_matrix}
	\begin{align}
		\tag{\ref{eq:cost_matrix}}
		C(\theta) &= \frac{1}{2T}\delta\mathbf{y}^\intercal \Sigma^{-1} \delta\mathbf{y}, \\
		\delta\mathbf{y} &\equiv \mathbf{y} - \mathbf{y}(\theta)
	\end{align}
\end{subequations}
where $\Sigma$ denotes the covariance matrix for the observation vector $\mathbf{y}$.
	\subsection{Sensitivity analysis and parameter classification}
	\label{sec:sensitivities}
	To quantify the sensitivity to parameters of model predictions at different time scales, we consider derivatives of the cost with respect to $\theta$.
Dropping the $t$ and $\theta$ dependence for clarity, the gradient of Eq.~\eqref{eq:cost_int} is
\begin{equation}
	\label{eq:gradient}
	\frac{\partial C}{\partial \theta_\mu} = -\frac{1}{T} \int_0^T dt \left(\delta y \cdot \frac{\partial y}{\partial \theta_\mu}\right),
\end{equation}
and the Hessian is
\begin{align}
  H_{\mu\nu} &\equiv \frac{\partial^2 C}{\partial \theta_\mu \partial \theta_\nu} \nonumber\\
		&= \frac{1}{T} \int_0^T dt \left(\frac{\partial y}{\partial \theta_\mu} \cdot \frac{\partial y}{\partial \theta_\nu} - \delta y \cdot \frac{\partial^2 y}{\partial \theta_\mu \partial \theta_\nu}\right).
\end{align}
Note that because $\delta y(t; \theta_0) = 0$, the gradient at $\theta_0$ is also 0 and the Hessian at $\theta_0$ simplifies to
\begin{equation}
	\label{eq:hessian}
  H_{\mu\nu}(\theta_0) = \frac{1}{T} \int_0^T dt \left(\frac{\partial y}{\partial \theta_\mu} \cdot \frac{\partial y}{\partial \theta_\nu}\right).
\end{equation}
This is also approximately valid when $\theta \approx \theta_0$.
For Eq.~\eqref{eq:cost_matrix}, the Hessian at $\theta_0$ takes the form
\begin{equation}
	\label{eq:FIM}
	H(\theta_0) = \frac{1}{T} \frac{\partial \mathbf{y}}{\partial \theta}^\intercal \Sigma^{-1} \frac{\partial \mathbf{y}}{\partial \theta}.
\end{equation}
Although the gradient and Hessian may be evaluated at other points, $H(\theta_0)$ is particularly important because it is the Fisher Information Metric (FIM) for this measurement process and acts as a Riemannian metric on the model manifold.
We are interested in the eigenvalues of $H$ and their dependence on $T$.

Figure \ref{fig:bigfig} plots a cross section of $C$ (as a surface over $\theta$), the eigenvalues of $H$, and a three-dimensional projection of the model manifold for three models (details of these models are found in Appendix \ref{sec:models}).
\begin{figure*}
	\includegraphics[width=\textwidth]{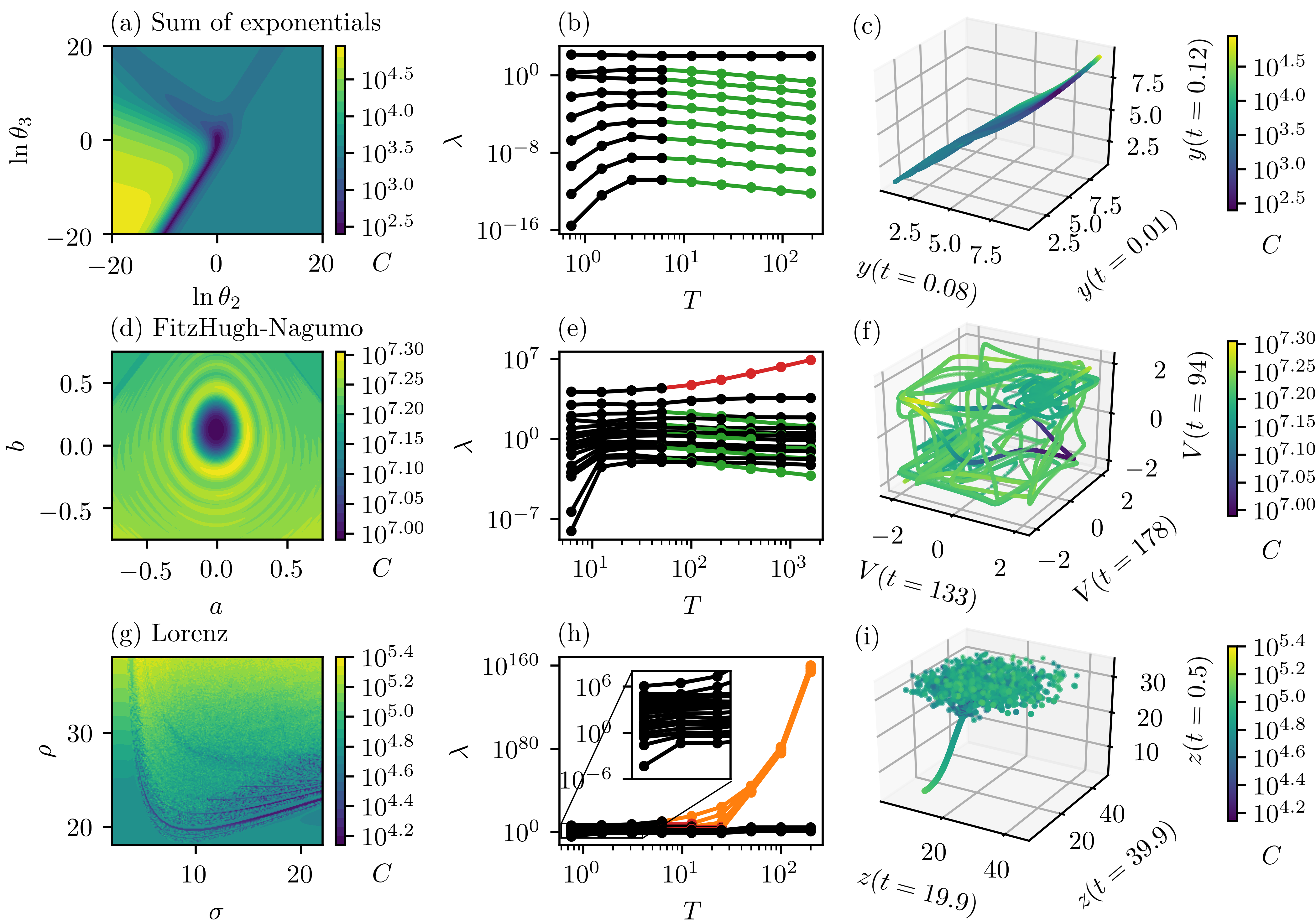}
	\caption{
		\textbf{Model classes.}
		\textit{\nth{1} column:} Cross sections of $C(\theta)$ [Eq.~\eqref{eq:cost_int}] for three prototype models (see Appendix \ref{sec:models}, models $A$, $D$, and $G$).
		Contrast the canyons in (a) with the ripples in (d) and the roughness in (g).
		\textit{\nth{2} column:} Hessian eigenvalues as a function of sampling time (same models).
		Colors differentiate scaling behaviors at long times.
		\textit{\nth{3} column:} Projections of the model manifold (same models).
		In (c), a 3-ball in parameter space was mapped to the nearly 1-dimensional region of prediction space shown (low effective dimensionality).
		By contrast, for (f) and (i) a single parameter was varied producing a 1D (space-filling) curve in prediction space (high effective dimensionality).
		Note: In (i), the model goes through a bifurcation where the manifold begins to oscillate rapidly.
		The sampling required to see continuity is prohibitive, so the points plotted become scattered.
		}
	\label{fig:bigfig}
\end{figure*}

The first model is characterized by a transient decay to a steady state.
As illustrated in Fig.~\ref{fig:bigfig}(b), for large $T$, the model becomes increasingly insensitive to parameter combinations that control transient behavior, scaling as $O(T^{-1})$.
The parameter that determines the steady state scales as $O(1)$.
These scaling behaviors can be motivated as follows.
We assume that parameter combinations which control the transient dynamics have sensitivities that decay to zero at long times,
\begin{equation}
	\frac{\partial y}{\partial \theta_{\mu}}(t\rightarrow\infty; \theta) \sim 0,
\end{equation}
while those that control the steady state are asymptotically constant:
\begin{equation}
	\frac{\partial y}{\partial \theta_{\mu}}(t\rightarrow\infty; \theta) \sim \text{const}.
\end{equation}
In light of Eq.~\eqref{eq:hessian}, this leads to the $O(T^{-1})$ and $O(1)$ scaling behaviors observed.
Note that as the total sampling time $T$ is increased past the transient dynamics, the only new information obtained is information about the final steady state.
Our choice of normalization keeps the effective number of data points constant, so increasing $T$ results in an effective loss of information about the transient dynamics but no information loss for the steady state.

The second model exhibits a periodic limit cycle.
As shown in Fig.~\ref{fig:bigfig}(e), parameter combinations controlling features of the attractor scale as $O(1)$, those that control the transient approach to the attractor scale as $O(T^{-1})$, and the combination controlling frequency scales as $O(T^2)$.
Motivation for the scaling behavior of the parameter combinations controlling the the transient approach to the attractor follow as in the previous case.
To motivate the other two scaling behaviors, we consider the steady state of the model and expand in a Fourier series:
\begin{align}
	y(t\rightarrow\infty; \theta) &= \sum_{k = -\infty}^{\infty}{c_k(\theta) e^{ik\omega(\theta)t}} \nonumber \\
		&= y(t\rightarrow\infty; c(\theta), \omega(\theta)).
\end{align}
There is an intermediate dependence of the steady state on the amplitude coefficients $c_k$ and the oscillatory frequency $\omega$.
This allows us to decompose the parameter sensitivities of the steady state into two parts:
\begin{equation}
	\label{eq:decoupling}
	\frac{\partial y}{\partial \theta_{\mu}}(t\rightarrow\infty; \theta) = \sum_{k = -\infty}^{\infty}{\frac{\partial y}{\partial c_k}\frac{\partial c_k}{\partial \theta_{\mu}}} + \frac{\partial y}{\partial \omega}\frac{\partial \omega}{\partial \theta_{\mu}}.
\end{equation}
Because $c_k$ and $\omega$ are time-independent by construction, the time dependence of these sensitivities is due entirely to the coefficients
\begin{equation}
	\frac{\partial y}{\partial c_k}(t\rightarrow\infty; \theta) = e^{ik\omega(\theta)t},
\end{equation}
which is bounded by a constant, and
\begin{align}
	\frac{\partial y}{\partial \omega}(t\rightarrow\infty; \theta) &= \sum_{k = -\infty}^{\infty}{iktc_k(\theta) e^{ik\omega(\theta)t}} \nonumber \\
		&\sim t,
\end{align}
which grows linearly with time.
The amplitude sensitivities $(\partial y/\partial c_k)(\partial c_k/\partial \theta_\mu)$ control the shape and amplitude of the steady state and give rise to $O(1)$ scaling behavior [referring again to Eq.~\eqref{eq:hessian}].
By contrast, the frequency sensitivity $(\partial y/\partial \omega)(\partial \omega/\partial \theta_\mu)$ results in $O(T^2)$ scaling behavior.
Other studies have focused on the sensitivity to period, rather than frequency, but the temporal scaling behavior is the same for both \cite{kramer1984sensitivity, larter1984sensitivity, wilkins2009sensitivity}.

Finally, the third model is chaotic; parameters controlling its dynamics exhibit exponential sensitivities,
\begin{equation}
	\frac{\partial y}{\partial \theta_{\mu}}(t\rightarrow\infty; \theta) \sim e^{\lambda_{\mu} t},
\end{equation}
leading to the exponential scaling behavior in Fig.~\ref{fig:bigfig}(h).

We classify parameter combinations in a model according to their scaling behavior; this classification is summarized in Table \ref{tab:classification}.
LaMont and Wiggins have also proposed a classification of model parameters, based on the \emph{complexity} of a given parameter combination \cite{lamont2016information}.
In the case of dynamical models, our analysis illustrates the mechanisms that give rise to the complexities of each class.
\begin{table}[htbp]
	\caption{
		\textbf{Parameter classification.}
		}
	\label{tab:classification}
	\begin{ruledtabular}
	\begin{tabular}{cc}
		Eigenvalue scaling behavior & Dynamics controlled \\ \hline
		$O(T^{-1})$                 & transient           \\
		$O(1)$                      & steady state        \\
		$O(T^2)$                    & frequency           \\
		$O(e^T)$                    & chaotic behavior    \\
	\end{tabular}
	\end{ruledtabular}
\end{table}
	\subsection{Model classification}
	\label{sec:model_classes}
	The different scaling behaviors for the Hessian eigenvalues are accompanied by different structures in both the cost surface and the model manifold (\nth{1} and \nth{3} columns of Fig.~\ref{fig:bigfig}).
The cost surface of the first model is characterized by a single, highly anisotropic basin.
Its model manifold is similarly anisotropic; the long, narrow \emph{hyper-ribbon} structure is common for models with low effective dimensionality \cite{transtrum2010nonlinear, transtrum2011geometry}.
In contrast, the second cost surface has many local minima and a model manifold with high curvature.
The third cost surface exhibits a fractal-like roughness (although for finite $T$ the derivative with respect to parameters formally exists everywhere).
Its model manifold is even more highly curved and space-filling.

These three models are prototypes of three model classes, distinguished by the scaling behavior of the largest eigenvalue for large $T$.
For the first class, $\lambda_{\text{max}} \sim O(1)$ is bound by a constant.
For the second class, $\lambda_{\text{max}} \sim O(T^n)$ is bound by a polynomial.
For the third class, $\lambda_{\text{max}} \sim O(e^T)$ is exponential.
We plot the eigenvalues of the Hessian (at large, fixed $T$) for the three protoype models and for two additional models from each class in Fig.~\ref{fig:eigvals} (details of these models are found in Appendix \ref{sec:models}).
All nine models are considered \emph{sloppy}; that is, the eigenvalues of the Hessian are spread over many orders of magnitude.
Accordingly, we refer to these model classes as \emph{sloppy models of the first, second, and third kinds}, respectively.
\begin{figure}[htbp]
	\includegraphics[width=1.0\columnwidth]{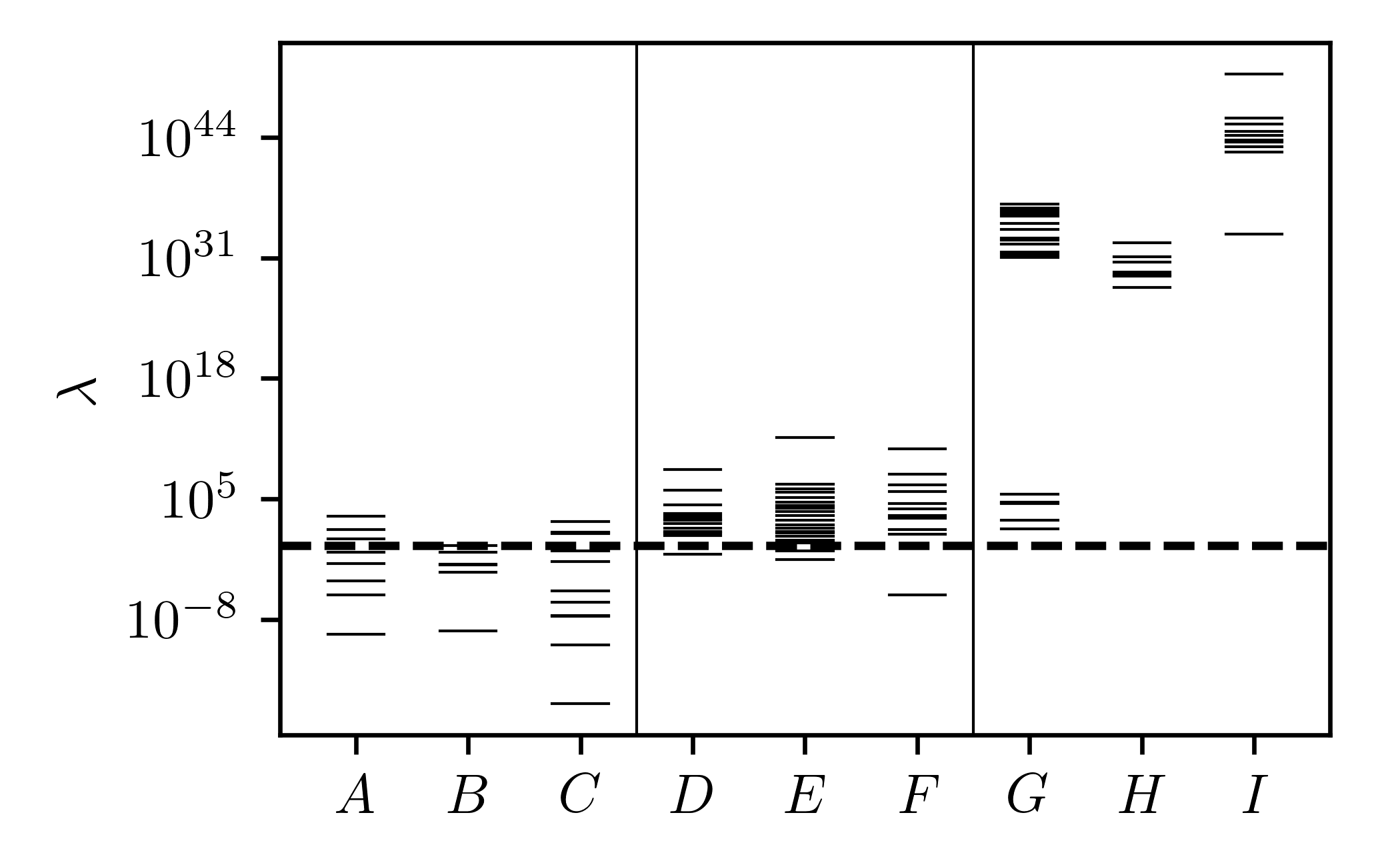}
	\caption{
		\textbf{Eigenvalues of $H(\theta_0)$} for the following models (see Appendix \ref{sec:models} for details):
		$A$) sum of exponentials;
		$B$) rational polynomial;
		$C$) biological adaptation;
		$D$) FitzHugh-Nagumo;
		$E$) Hodgkin-Huxley;
		$F$) Wnt oscillator;
		$G$) Lorenz;
		$H$) Hindmarsh-Rose;
		$I$) damped, driven pendulum.
		$\{A, B, C\}$ are nonoscillatory models, $\{D, E, F\}$ are periodic, and $\{G, H, I\}$ are chaotic.
		}
	\label{fig:eigvals}
\end{figure}

\section{Manifold curvature}
\label{sec:curvature}
The large sensitivities of sloppy models of the second and third kinds are necessarily associated with large curvature and high effective dimensionality on the model manifold.
This can be understood by noting that the absolute variation in the model behavior is bounded (the models oscillate within a finite range and do not grow).
This restricts the model manifold to a finite region of data space.
Large parameter sensitivities indicate that the model manifold is very long in the associated parameter directions.
The only way to fit something very long into a finite region is for it to curve, fold, or wind.
The combination of large parameter sensitivities and bounded predicted behavior necessitates large manifold curvature.
For large $T$, there will be enough winding that the manifold effectively fills a volume of higher dimension than that of the manifold itself.
This high effective dimensionality is the opposite effect of the low effective dimensionality argued for in previous studies of sloppy models \cite{machta2013parameter, transtrum2015perspective}.

To quantify this effect, the extrinsic curvature associated with parameter direction $v$ is given by the geodesic curvature $k(v) = 1/R$ (as in Ref.~\cite{transtrum2011geometry}), where $R$ is the radius of curvature of the circle tangent to the manifold along direction $(\partial \mathbf{y}/\partial \theta^\mu) v^\mu$ (sum over $\mu$ implied; see Fig.~\ref{fig:tang_cir}).
We define the \emph{winding frequency},
\begin{equation}
	\label{eq:wind_freq}
  \omega(v) \equiv \left\vert \frac{\partial \mathbf{y}}{\partial \theta^\mu} v^\mu \right\vert k(v)
\end{equation}
which is the angular velocity at which the manifold winds around the tangent circle, such that $f = \omega/2\pi$ is the number of windings of the manifold per unit change in parameters.
Because $C$ is a distance measure for the manifold embedding space, each winding of the manifold results in a local minimum of $C$, so $f$ is also the frequency of local minima in $C$ as we move along parameter direction $v$.
\begin{figure}[htbp]
	\includegraphics[width=1.0\columnwidth]{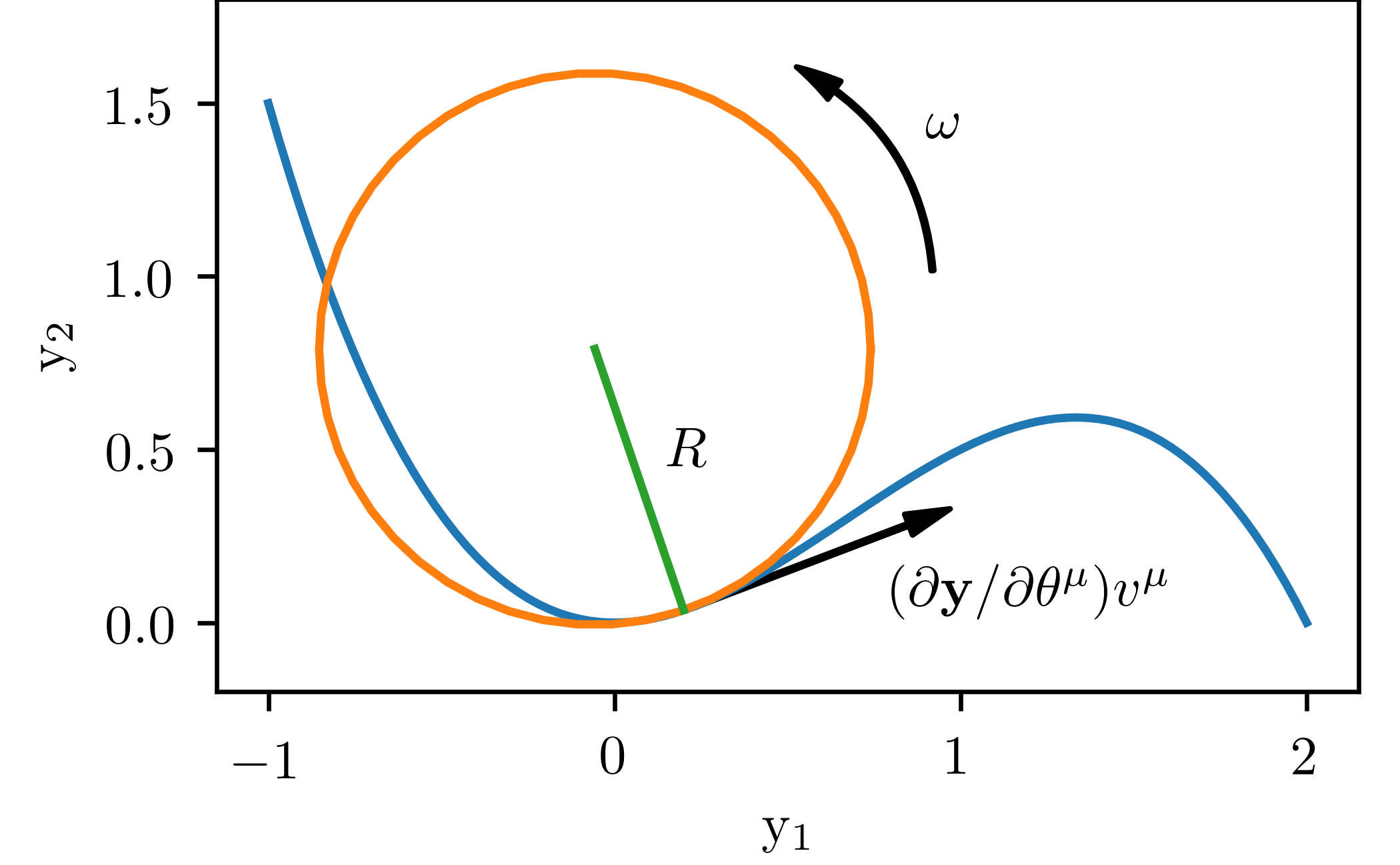}
	\caption{
		\textbf{Illustration of winding frequency.}
		The ``s''-shaped curve represents a possible 1D cross-section of a model manifold (obtained, for example, by varying just one parameter combination) in a simple 2D data space.
		Also shown are the tangent/velocity vector $(\partial \mathbf{y}/\partial \theta^\mu) v^\mu$ (sum over $\mu$ implied), the tangent circle with radius $R$, and the winding frequency $\omega$ defined in Eq.~\eqref{eq:wind_freq}.
		}
	\label{fig:tang_cir}
\end{figure}

Figure \ref{fig:freqs} shows winding frequencies along Hessian eigendirections for the models from Fig.~\ref{fig:eigvals}.
Notice that sloppy models of the first kind (i.e., hyper-ribbons) have low winding frequencies.
Sloppy models of the second kind have high winding frequency in the stiffest direction, which controls frequency.
Sloppy models of the third kind have high winding frequencies in more than one direction.
\begin{figure}[htbp]
	\includegraphics[width=1.0\columnwidth]{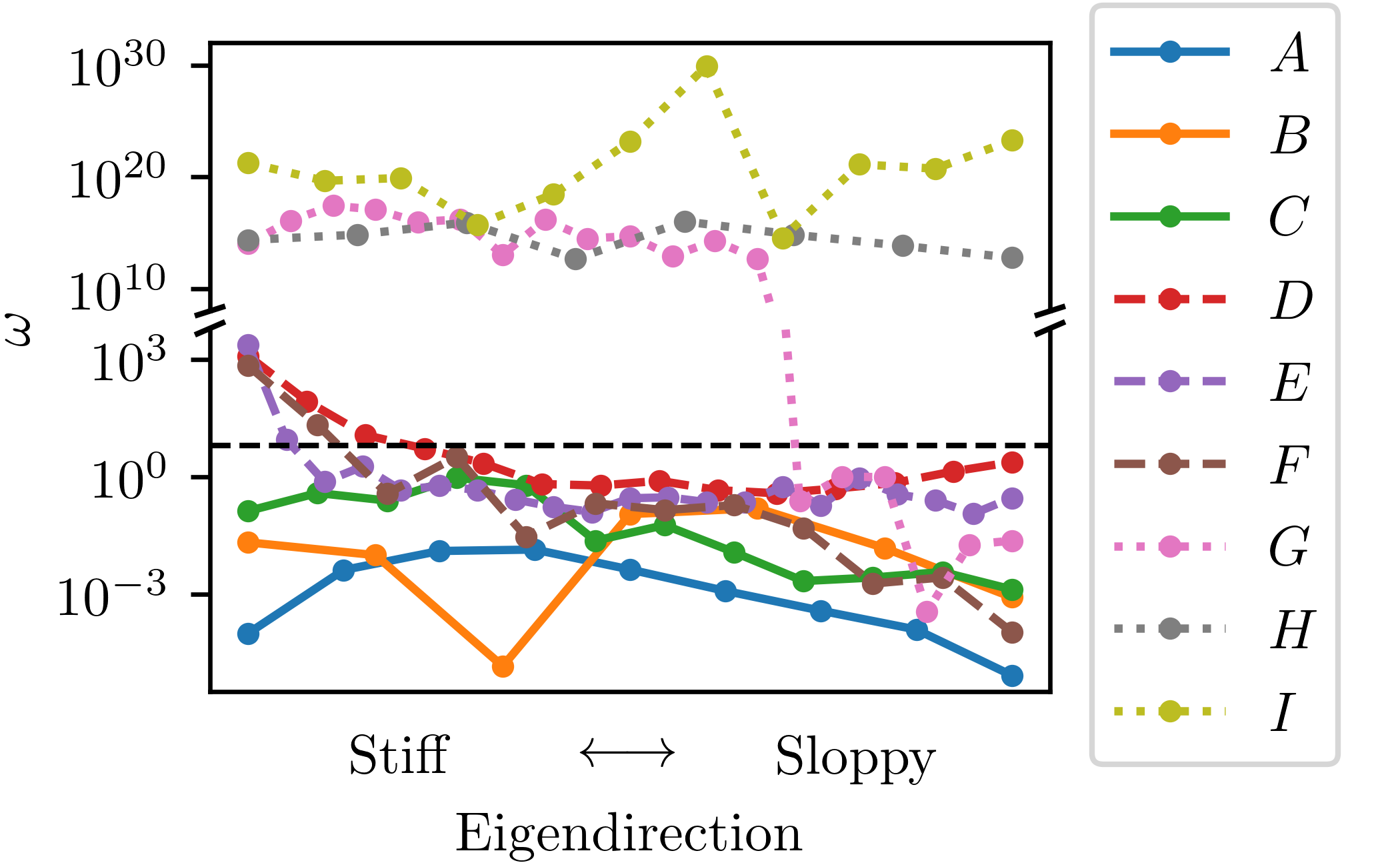}
	\caption{
		\textbf{Winding frequencies} along Hessian eigendirections for the models from Fig.~\ref{fig:eigvals}, ordered from left to right by magnitude of the corresponding eigenvalue.
		``Stiff'' refers to eigendirections with large eigenvalues, while ``sloppy'' refers to eigendirections with small eigenvalues.
		The black dashed line at $\omega = 2\pi$ roughly distinguishes low from high winding frequencies.
		}
	\label{fig:freqs}
\end{figure}

The effective dimensionality, estimated by the winding frequencies, depends on the metric of the model manifold embedding space, i.e., Eq.~\eqref{eq:cost_int}.
We now show that alternative choices for embedding the model manifold can lead to different effective dimensionalities.

\section{Alternative metrics}
\label{sec:metrics}
	\subsection{Analytic signal (AS)}
	\label{sec:AS}
	The high effective dimensionality of sloppy models of the second kind is due entirely to the parameter combination controlling frequency. 
Varying this parameter combination causes model predictions to pass in and out of phase with each other, resulting in local minima in the cost (see Figs.~\ref{fig:cost_dvt} and \ref{fig:time_vs_phase}).
We avoid this aliasing by defining the phase of oscillation as a monotonically increasing function of time and comparing model behaviors at the same phase.

\begin{figure}[htbp]
	\includegraphics[width=1.0\columnwidth]{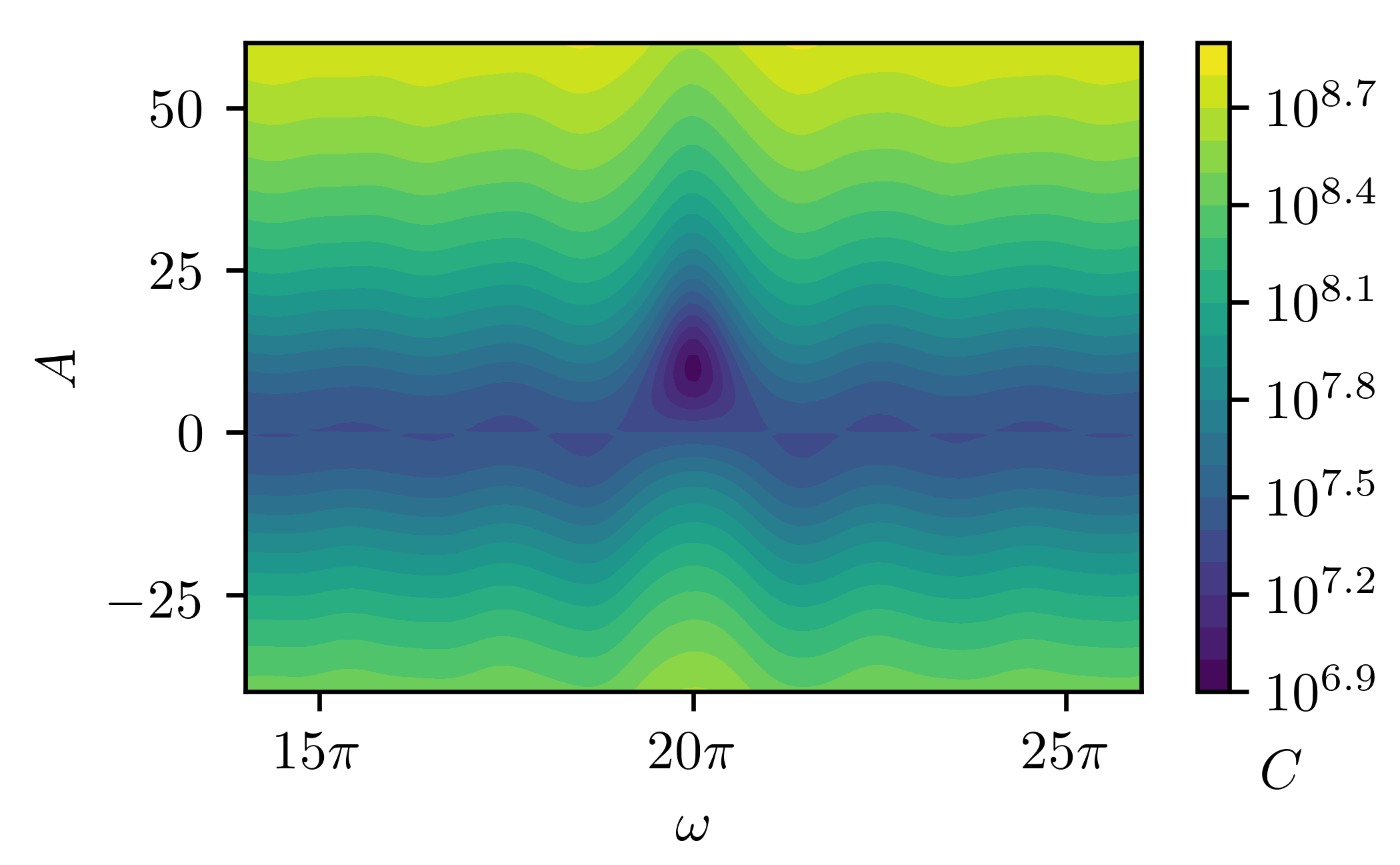}
	\caption{
		\textbf{Cost $C$} for the model $y(t) = A\cos(\omega t)$, treating $A$ and $\omega$ as parameters.
		The cost has been rescaled to make the local minima apparent.
		}
	\label{fig:cost_dvt}
\end{figure}

\begin{figure}[htbp]
	\includegraphics[width=1.0\columnwidth]{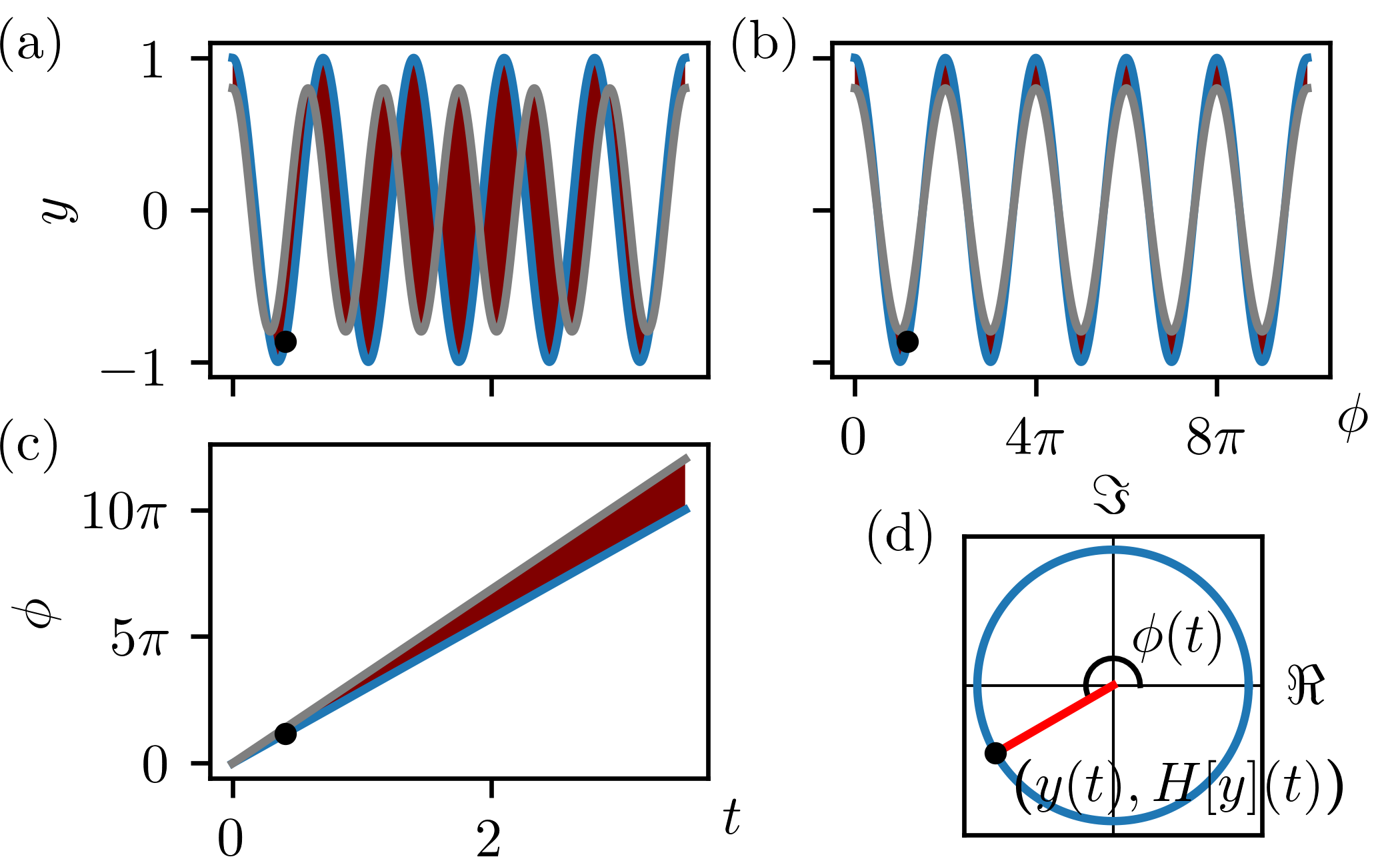}
	\caption{
		\textbf{Decoupling amplitude from phase.}
		(a) Signal vs. time for two signals with mismatched amplitude and frequency; their difference is indicated by the shading between the curves.
		The mismatch in frequency causes a large difference $\delta y$ when the two signals are out of phase ($t \approx 1.8$) but little or no difference when they are in phase ($t \approx 3.5$).
		(b) Signal vs.~phase for the same signals.
		The difference-at-constant-phase $\delta \tilde{y}$ is consistent throughout (see Sec.~\ref{sec:AS_cost}).
		(c) Phase vs.~time for the two signals.
		The difference in phase $\delta \phi$ simply grows linearly (see Sec.~\ref{sec:AS_cost}).
		(d) Analytic representation in the complex plane of the black point marked in the other 3 panels (see Sec.~\ref{sec:phase_def}).
		}
	\label{fig:time_vs_phase}
\end{figure}

Many definitions of instantaneous frequency and phase have been considered in the literature \cite{boashash1992estimating, pikovsky1997phase, kiss2001phase}.
We use the \emph{analytic signal} approach \cite{gabor1946theory}, which is discussed in Sec.~\ref{sec:phase_def}.
Some alternatives are discussed in Appendix \ref{sec:alt_phase}.
We propose a new metric for oscillatory systems in Sec.~\ref{sec:AS_cost}.
We discuss the Hessian and FIM in Sec.~\ref{sec:FIM}.
Results of applying the new metric to the FitzHugh-Nagumo model are found in Sec.~\ref{sec:AS_results}.
Comparing model predictions with observational data in the new paradigm requires calculating the phases of the observations, which will have uncertainty.
We propagate uncertainty and derive appropriate covariance matrices in Appendix \ref{sec:covariance}.
Calculation of winding frequencies requires second-order parameter sensitivities (specifically, when calculating the geodesic curvature $\kappa$); we derive the necessary formulas in Appendix \ref{sec:derivatives}.
		\subsubsection{Phase definition}
		\label{sec:phase_def}
		The analytic representation $z(t)$ of an oscillatory signal $y(t)$ is a complex function defined as
\begin{equation}
	\label{eq:AS}
	z(t) \equiv y(t) + i H[y](t) = A(t)e^{i\phi(t)},
\end{equation}
where $H[y](t)$ is the Hilbert transform of $y(t)$,
\begin{equation}
	H[y](t) \equiv \frac{1}{\pi}\PV\int_{-\infty}^{\infty}{\frac{y(\tau)}{t - \tau}d\tau},
\end{equation}
and the magnitude $A(t)$ and argument $\phi(t)$ of $z(t)$ are
\begin{align}
	\label{eq:AS_amp}
	A(t) & \equiv \sqrt{y^2(t) + H^2[y](t)} \\
	\label{eq:AS_phase}
	\phi(t) & \equiv \tan^{-1}\left(\frac{H[y](t)}{y(t)}\right).
\end{align}
[See Fig.~\ref{fig:time_vs_phase}(d).]

In light of Eq.~\eqref{eq:AS}, we reinterpret $y(t)$ in terms of amplitude and phase as
\begin{equation}
	\label{eq:amp_and_phase}
	y(t) = \Re\left\{z(t)\right\} = A(t)\cos\bm{(}\phi(t)\bm{)}.
\end{equation}
We then define a new signal $\tilde{y}$ as a function of phase:
\begin{equation}
	\label{eq:time_to_phase}
	\tilde{y}\bm{(}\phi(t)\bm{)} \equiv y(t) = A(t)\cos\bm{(}\phi(t)\bm{)}.
\end{equation}
As long as $\phi(t)$ is monotonically increasing, the relationship between $\phi$ and $t$ is invertible.
Hence, we may also write
\begin{equation}
	\label{eq:time_to_phase2}
	\tilde{y}(\phi) = A\bm{(}t(\phi)\bm{)}\cos(\phi).
\end{equation}

If $y(t)$ is a vector (rather than scalar) function of time, then $\phi(t)$ will also be a vector function of time.
That is, for each scalar component of $y(t)$, the preceding prescription for constructing the phase should be applied separately.
If this is not possible or does not produce a set of monotonically increasing phases, it may be applied to a single scalar component of $y$ and the resulting phase used for all of the components.
For other alternatives that avoid using the Hilbert transform, see Appendix \ref{sec:alt_phase}.

As a final note, a necessary condition for $\phi(t)$ to be monotonically increasing is that the signal $y(t)$ oscillate around 0.
If it is not, the time average $\left\langle y(t) \right\rangle = (1/T)\int_0^T{y(t)dt}$ should be subtracted from $y(t)$ prior to calculating the phase.
$H[y](t)$ will be unaffected, as the Hilbert transform of a constant is 0.
		\subsubsection{New cost using phase}
		\label{sec:AS_cost}
		We want to construct a cost that compares models at the same phase rather than the same time.
Actually, we can go one step further and construct a cost that also retains the phase information while still eliminating the aliasing of oscillations that results in local minima.
We use an approximation of Eq.~\eqref{eq:cost_int_devs} that arises from the propagation of uncertainty considered in Appendix \ref{sec:covariance} [see Eq.~\eqref{eq:devs_ytilde}].
In the discrete case [comparing a model $y(t_i; \theta)$ with observational data $y_i$], Eq.~\eqref{eq:cost_int_devs} is
\begin{equation}
	\label{eq:devs_t}
	\delta y_i \equiv y_i - y(t_i; \theta).
\end{equation}
We define the deviations of the phases $\phi_i$ of the observations from the phases $\phi(t_i; \theta)$ predicted by the model as
\begin{equation}
	\label{eq:devs_phi_def}
	\delta \phi_i \equiv \phi_i - \phi(t_i; \theta)
\end{equation}
[see Fig.~\ref{fig:time_vs_phase}(c)],
and the deviations of the observations from the predictions at constant phase as
\begin{equation}
	\label{eq:devs_ytilde_def}
	\delta \tilde{y}_i \equiv y_i - \tilde{y}(\phi_i; \theta)
\end{equation}
[see Fig.~\ref{fig:time_vs_phase}(b)].
The approximation we use for oscillatory systems is
\begin{equation}
	\label{eq:deviations}
	\delta y_i \approx \delta \tilde{y}_i + \left(\frac{\partial \tilde{y}}{\partial \phi}\right)_i\delta \phi_i \equiv \delta \hat{y}_i.
\end{equation}
The first term captures changes in amplitude while the second term captures changes in phase or frequency, so both pieces of information are retained (see Sec.~\ref{sec:FIM}).
However, because this approximation is first order in $\delta \phi_i$, it does, in fact, eliminate the nonlinear dependence on frequency that results in ripples in the cost (refer back to Fig.~\ref{fig:time_vs_phase}), which we will show.

We define a new cost function by replacing $\delta\mathbf{y}$ in Eq.~\eqref{eq:cost_matrix} with the approximation $\delta\hat{\mathbf{y}}$ defined according to Eq.~\eqref{eq:deviations}:
\begin{subequations}
	\label{eq:cost_phase}
	\begin{equation}
		\tag{\ref{eq:cost_phase}}
		C^\phi(\theta) \equiv \frac{1}{2T}\delta \hat{\mathbf{y}}^\intercal \Sigma^{-1} \delta \hat{\mathbf{y}}
	\end{equation}
Using Eq.~\eqref{eq:deviations}, this may be decomposed into three pieces representing the amplitude contribution, the phase contribution, and a cross term:
	\begin{align}
		\label{eq:cost_decomp}
		C^\phi(\theta) &= C^{\tilde{y}(\phi)}(\theta) + C^{\phi(t)}(\theta) + C^{X}(\theta) \\
		C^{\tilde{y}(\phi)}(\theta) &\equiv \frac{1}{2T}\delta \tilde{\mathbf{y}}^\intercal \Sigma^{-1} \delta \tilde{\mathbf{y}} \\
		C^{\phi(t)}(\theta) &\equiv \frac{1}{2T}\Delta \bm{\phi}^\intercal \Sigma^{-1} \Delta \bm{\phi} \\
		C^X(\theta) &\equiv \frac{1}{2T}\left(\delta \tilde{\mathbf{y}}^\intercal \Sigma^{-1} \Delta \bm{\phi} + \Delta \bm{\phi}^\intercal \Sigma^{-1} \delta \tilde{\mathbf{y}}\right) \\
		\Delta \phi_i &\equiv \left(\frac{\partial \tilde{y}}{\partial \phi}\right)_i\delta \phi_i.
	\end{align}
\end{subequations}
We compare $C$ [Eq.~\eqref{eq:cost_matrix}], $C^{\tilde{y}(\phi)}$, $C^{\phi(t)}$, and $C^{\phi}$ in Fig.~\ref{fig:sin_cost_decomp}.

\begin{figure}[htbp]
	\includegraphics[width=1.0\columnwidth]{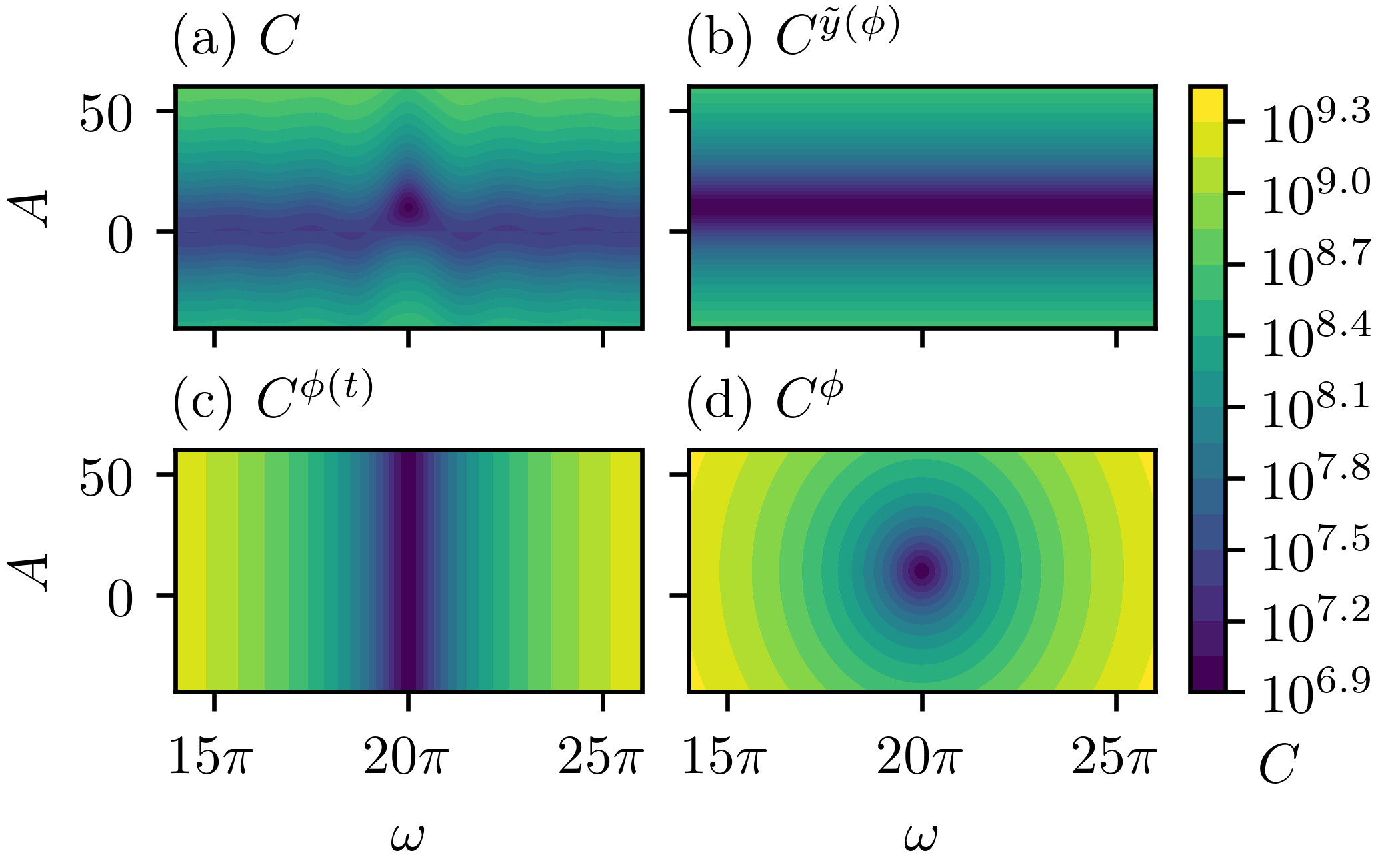}
	\caption{
		\textbf{Cost decomposition} of the model $y(t) = A\cos(\omega t)$.
		(a) Same as Fig.~\ref{fig:cost_dvt}.
		(b) $\tilde{y}(\phi) = A\cos(\phi)$ is insensitive to changes in $\omega$ and varies linearly with $A$, resulting in a quadratic dependence of $C^{\tilde{y}(\phi)}$ on $A$ only.
		(c) $\phi(t) = \omega t$ is insensitive to changes in $A$ and varies linearly with $\omega$, resulting in a quadratic dependence of $C^{\phi(t)}$ on $\omega$ only.
		(d) Cost using Eq.~\eqref{eq:cost_phase}.
		The ripples in (a) have been replaced with a quadratic basin.
		}
	\label{fig:sin_cost_decomp}
\end{figure}

When comparing two models with parameters $\theta_0$ and $\theta$, Eqs.~\eqref{eq:devs_phi_def}, \eqref{eq:devs_ytilde_def}, \eqref{eq:deviations}, and \eqref{eq:cost_phase} take the form
\begin{align}
	\delta \phi(t; \theta) &\equiv \phi(t; \theta_0) - \phi(t; \theta), \\
	\delta \tilde{y}(t; \theta) &\equiv \tilde{y}\bm{(}\phi(t; \theta_0); \theta_0\bm{)} - \tilde{y}\bm{(}\phi(t; \theta_0); \theta\bm{)} \nonumber\\
		&= y(t; \theta_0) - \tilde{y}\bm{(}\phi(t; \theta_0); \theta\bm{)}, \\
	\delta \hat{y}(t; \theta) &\equiv \delta \tilde{y}(t; \theta) + \frac{\partial \tilde{y}\bm{(}\phi(t; \theta_0); \theta_0\bm{)}}{\partial \phi}\delta \phi(t, \theta),
\end{align}
\begin{equation}
	\label{eq:cost_AS}
	C^\phi(\theta) \equiv \frac{1}{2T} \int_0^T dt \bm{(}\delta \hat{y}(t; \theta) \cdot \delta \hat{y}(t; \theta)\bm{)}.
\end{equation}
As we show in Sec.~\ref{sec:FIM}, Eq.~\eqref{eq:cost_AS} is a quadratic approximation of Eq.~\eqref{eq:cost_int} (i.e., they have the same gradient and Hessian).
In other words, Eq.~\eqref{eq:cost_AS} is an isometric embedding of the model manifold.
However, because changes in frequency only affect $\phi(t; \theta)$, which is unbounded, the large manifold volume is no longer constrained to a finite region of the embedding space.
		\subsubsection{Fisher Information Metric}
		\label{sec:FIM}
		We stated in Sec.~\ref{sec:sensitivities} that the Hessian of the cost evaluated at $\theta_0$ is the Fisher Information Metric (FIM).
Specifically, the FIM is related to the cost by
\begin{equation}
	\label{eq:FIM_def}
	\mathcal{I}_{\mu\nu} = \left<\frac{\partial^2 C(\theta_0)}{\partial \theta_{\mu} \partial \theta_{\nu}}\right> = \left< H_{\mu\nu}(\theta_0) \right>.
\end{equation}
We have already shown that
\begin{equation}
	\label{eq:FIM_int}
	\mathcal{I}_{\mu\nu} = \frac{1}{T} \int_0^T dt \frac{\partial y}{\partial \theta_\mu} \cdot \frac{\partial y}{\partial \theta_\nu}
\end{equation}
for Eq.~\eqref{eq:cost_int} [see Eq.~\eqref{eq:hessian}].
We can rewrite this for oscillatory systems in light of Eq.~\eqref{eq:time_to_phase},
\begin{equation}
	\label{eq:time_to_phase3}
	y(t; \theta) = \tilde{y}\bm{(}\phi(t; \theta); \theta\bm{)}.
\end{equation}
Differentiating Eq.~\eqref{eq:time_to_phase3} with respect to $\theta_\mu$, we obtain
\begin{equation}
	\label{eq:decoupling2}
	\left.\frac{\partial y}{\partial \theta_\mu}\right|_t = \left.\frac{\partial \tilde{y}}{\partial \theta_\mu}\right|_\phi + \left.\frac{\partial \tilde{y}}{\partial \phi}\right|_\theta \left.\frac{\partial \phi}{\partial \theta_\mu}\right|_t,
\end{equation}
where the $|_x$ notation is used to indicate that the argument $x$ is being held constant in the given derivative.
This relationship is exact and shows a decoupling of the amplitude sensitivity from the phase sensitivity [similar to Eq.~\eqref{eq:decoupling}].
Substituting Eq.~\eqref{eq:decoupling2} into Eq.~\eqref{eq:FIM_int} yields
\begin{equation}
	\label{eq:FIM_AS}
	\mathcal{I}_{\mu\nu} = \frac{1}{T} \int_0^T dt \left(\frac{\partial \tilde{y}}{\partial \theta_\mu} + \frac{\partial \tilde{y}}{\partial \phi} \frac{\partial \phi}{\partial \theta_\mu}\right) \cdot \left(\frac{\partial \tilde{y}}{\partial \theta_\nu} + \frac{\partial \tilde{y}}{\partial \phi} \frac{\partial \phi}{\partial \theta_\nu}\right).
\end{equation}

We now show that Eq.~\eqref{eq:FIM_AS} is also the FIM for Eq.~\eqref{eq:cost_AS}.
First we calculate the gradient of $C^{\phi}(\theta)$:
\begin{equation}
	\frac{\partial C^\phi}{\partial \theta_\mu} = -\frac{1}{T} \int_0^T dt \left\{\delta \hat{y} \cdot \left(\frac{\partial \tilde{y}}{\partial \theta_\mu} + \frac{\partial \tilde{y}}{\partial \phi} \frac{\partial \phi}{\partial \theta_\mu}\right)\right\}.
\end{equation}
Next we calculate the Hessian and evaluate it at $\theta_0$ [note that $\delta\hat{y}(\theta_0) = \delta\phi(\theta_0) = 0$]:
\begin{equation}
	H_{\mu\nu}(\theta_0) = \frac{1}{T} \int_0^T dt \left(\frac{\partial \tilde{y}}{\partial \theta_\mu} + \frac{\partial \tilde{y}}{\partial \phi} \frac{\partial \phi}{\partial \theta_\mu}\right) \cdot \left(\frac{\partial \tilde{y}}{\partial \theta_\nu} + \frac{\partial \tilde{y}}{\partial \phi} \frac{\partial \phi}{\partial \theta_\nu}\right).
\end{equation}
Clearly this is the same as Eq.~\eqref{eq:FIM_AS}.
Because the FIM is preserved, the new cost [Eq.~\eqref{eq:cost_AS}] constitutes an isometric embedding of the model manifold and no information is lost (in the sense of the Fisher Information).
		\subsubsection{Results}
		\label{sec:AS_results}
		We implement the new metric for the FitzHugh-Nagumo model as an example; results are shown in Fig.~\ref{fig:resultsfig}(a-c).
The local minima in the cost surface in Fig.~\ref{fig:bigfig}(d) have been eliminated, the winding frequency of the stiffest direction is significantly reduced, and the manifold is no longer highly curved (see also Appendix \ref{sec:convergence}).
Because the new cost function is an isometric embedding [i.e., preserves the Hessian in Eq.~\eqref{eq:hessian}], the curvature of the cost surface at the bottom of the bowl is the same as that in Fig.~\ref{fig:bigfig}(d).

\begin{figure*}
	\includegraphics[width=\textwidth]{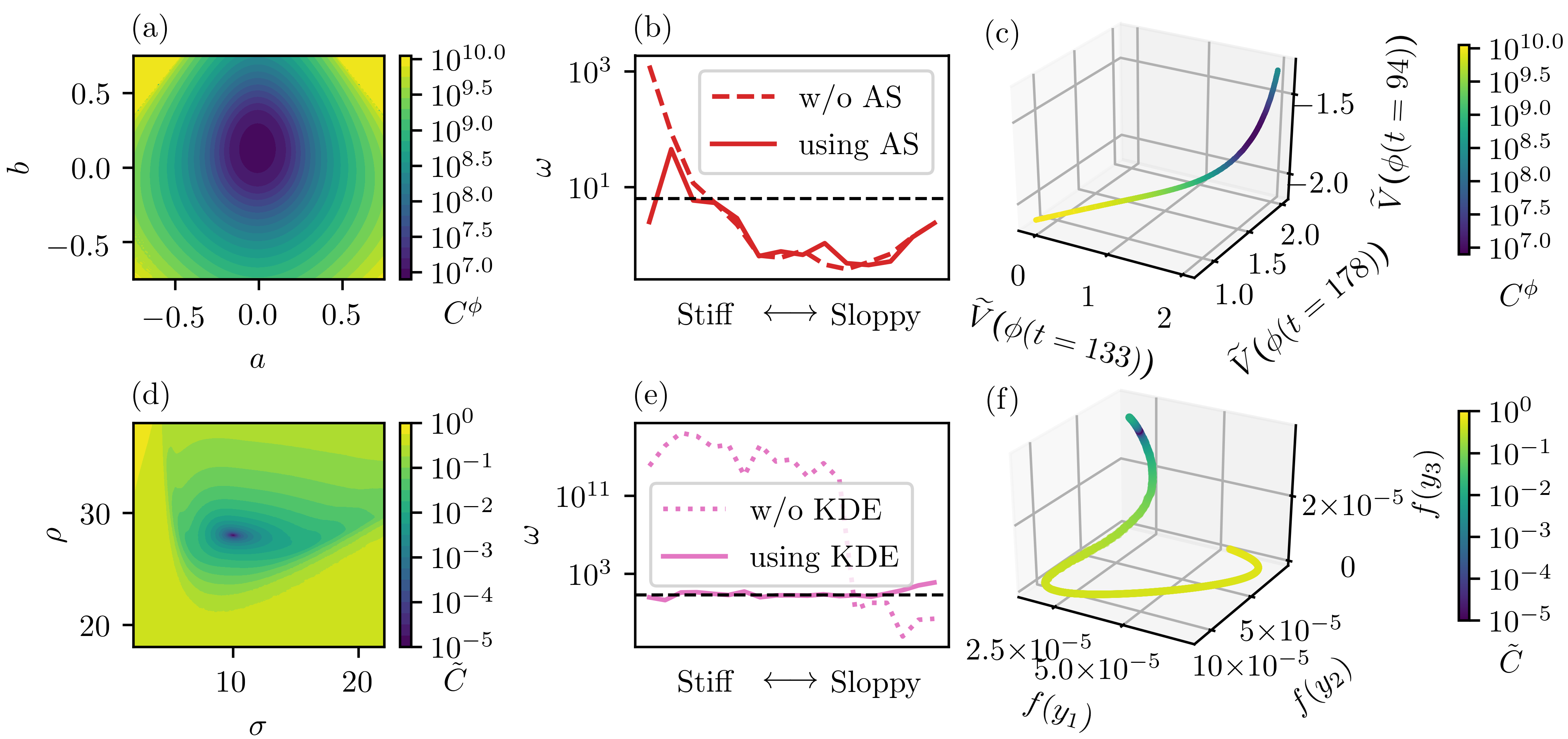}
	\caption{
		\textbf{Effects of alternative metrics} on cost surfaces (\nth{1} column), winding frequencies (\nth{2} column), and manifolds (\nth{3} column), (a-c) using analytic signal (AS) and (d-f) using kernel density estimation (KDE).
		Compare (a), (c), (d), and (f) with Fig.~\ref{fig:bigfig}(d), (f), (g), and (i), respectively.
		(b) and (e) show both the winding frequencies shown previously in Fig.~\ref{fig:freqs} (``w/o \_\_\_'') and the winding frequencies that result when using the new metrics (``using \_\_\_'') for comparison.
		Note that long time series are needed to achieve these results; see Appendix \ref{sec:convergence} for details.
		}
	\label{fig:resultsfig}
\end{figure*}
	\subsection{Kernel density estimation (KDE)}
	\label{sec:KDE}
	The high effective dimensionality of sloppy models of the third kind cannot be eliminated using the new metric discussed in Sec.~\ref{sec:AS}.
Adjusting the phase of a chaotic time series is insufficient to account for the variation in predictions as one moves from point to point in parameter space [resulting in the apparent roughness of the cost surface illustrated in Fig.~\ref{fig:bigfig}(g)].
This is reflected in the exponential sensitivities of chaotic systems at long times and is connected with a fundamental difference in manifold structure between sloppy models of the second and third kinds.
Note from Eq.~\eqref{eq:wind_freq} that winding frequency is directly proportional to geodesic/extrinsic curvature.
Figure \ref{fig:freqs} shows that the manifolds of sloppy models of the second kind only have high (extrinsic) curvature in one direction (like a scroll of paper).
This high curvature can be eliminated through an isometric embedding (analogous to unwinding the scroll).
By contrast, sloppy models of the third kind have high extrinsic curvature in more than one dimension.
High extrinsic curvature in multiple dimensions is necessarily associated with some intrinsic curvature, and this intrinsic curvature cannot be eliminated through an isometric embedding (think of a globe, which can't be ``unwound'' and laid flat).

The sensitivities of chaotic time series to parameters (including initial conditions) make time series prediction in sloppy models of the third kind impractical at long times.
However, model predictions $y_i(\theta)$ in phase space do give rise to a predictable distribution $f(y, \theta)$ \cite{lasota2013chaos}.
We evolve an ensemble of initial conditions and use the result to approximate this distribution with a \emph{kernel density estimate} \cite{rosenblatt1956remarks, parzen1962estimation},
\begin{equation}
	\label{eq:kde}
	f(y, \theta) \approx \frac{1}{nh}\sum_{i = 1}^n{K\left(\frac{y - y_i(\theta)}{h}\right)},
\end{equation}
where $n$ is the number of predictions/observations, $K(\cdot)$ is a kernel function, and $h$ is the kernel bandwidth.
A natural metric to use for distributions is the \emph{Hellinger distance},
\begin{equation}
	\label{eq:cost_KDE}
	\tilde{C}(\theta) \equiv \frac{1}{2} \int dy \left(\sqrt{f(y, \theta_0)} - \sqrt{f(y, \theta)}\right)^2,
\end{equation}
because it is a quadratic form, which can be interpreted as a Euclidean embedding space.
It also induces a metric on the model manifold that is given by the Fisher Information Metric.

We implement this cost for the Lorenz system; results are shown in Fig.~\ref{fig:resultsfig}(d-f).
The ``rough'' cost surface of Fig.~\ref{fig:bigfig}(g) has been replaced with a basin, the high winding frequencies have all disappeared, and the manifold is regular (see also Appendix \ref{sec:convergence}).

\section{Conclusion}
\label{sec:conclusion}
Multimodality in comparing and training multi-parameter models is a common problem \cite{mendes1998non, moles2003parameter, annan2004efficient, lorenc20074d, ramsay2007parameter}.
Many common search algorithms find only a local minimum (a point in parameter space which locally minimizes the chosen distance measure) and not the global one.
Even with global search methods the possibility of local minima reduces confidence that the global minimum will be found.
Here, we have shown how the choice of distance measure affects the number of local minima.
We have quantified this effect using curvature on the model manifold and introduced the \emph{winding frequency} to estimate the density of local minima in parameter space.
Finally, we have shown that through an appropriate choice of metric, the model manifold can be systematically ``unwound'' to remove local minima while preserving relevant physical interpretations of distance.

In this paper we have studied systems for which the relevant structures were known a priori (e.g., limit cycles and strange attractors).
However, the metrics we propose may also be useful for identifying previously unknown structures in other complex systems.

One of the ongoing challenges for many complex systems is the development of appropriate reduced-order representations \cite{anderson1972more, theodoropoulos2000coarsstabil, machta2013parameter, transtrum2016bridging}.  
More broadly, it has been suggested that the existence of useful simplified models is often due to a systematic compression of parameter space \cite{machta2013parameter}.
Compressing the parameter space leads to a type of ``universality class'' in which models with different parameter values make statistically indistinguishable predictions.
This line of work has also lead to new methods for constructing simplified models from more complex and complete mechanistic representations \cite{transtrum2014model}.
Ultimately, this compression is a consequence of the similarity metric used to compare model behaviors.  

For sloppy models of the first kind (which have previously dominated the literature), the compression ``squashes'' some dimensions to be very thin (as in Fig.~\ref{fig:bigfig}(a) and Refs.~\cite{transtrum2010nonlinear, transtrum2011geometry}), leading to a universality class of continuously-connected parameters for which reduced-order models can be systematically derived \cite{transtrum2014model}.
In contrast, for sloppy models of the second and third kind, the manifold is wound tightly, so that a compression leads to a manifold folding in which non-contiguous regions of the manifold are identified as part of the same universality class.
It is unlikely that predictive reduced-order models can be found for sloppy models with high winding frequencies as this would imply the existence, for example, of accurate long-term weather forecasts.
High winding frequency is the information-geometric equivalent of sensitivity to microscopic details (such as frequency, initial conditions, or other parameters), well-studied in chaotic systems.
In contrast, by unwinding the model manifold using an alternative metric, the manifold is transformed into a hyper-ribbon and this extreme-sensitivity to microscopic details is removed.

Understanding how effective theories emerge at long time scales is a challenging problem that has drawn on sophisticated expertise from a variety of fields, including dynamical systems \cite{gear2003equatfree, desroches2012mixed}, signal processing \cite{kramer1984sensitivity, wilkins2009sensitivity}, statistics \cite{lamont2017information, mattingly2018maximizing}, and optimization \cite{ramsay2007parameter, lorenc20074d}.
In this work we have combined insights from these other domains with tools of information geometry.
Our hope is that this explicit connection will bring new tools, such as sloppy model analysis and the manifold boundary approximation method, to bear on a wide range of important, ongoing scientific problems.

\begin{acknowledgments}
We thank A. Stankovic for helpful comments.
Resources at the Fulton Supercomputing Lab at Brigham Young University were used for many of the computations.
BLF and MKT were supported by the National Science Foundation through grant NSF EPCN-1710727.
\end{acknowledgments}

\bibliographystyle{aps}
\bibliography{AS_bib}

\newpage
\appendix{}

\section{Models}
\label{sec:models}
Following are the models examined in Figs.~\ref{fig:bigfig}, \ref{fig:eigvals}, \ref{fig:freqs}, and \ref{fig:resultsfig}.
In some cases, additional polynomial terms (with parameters as coefficients) were added to the model equations of motion.
This allows calculation of the \emph{structural susceptibility} of the model, that is, susceptibility to perturbations of the underlying dynamics \cite{chachra2012structural}.
These terms can be thought of as representing details of the real system that have been left out of the model.
\begin{enumerate}[label=\emph{\Alph*})]
	\item A sum of decaying exponentials leading to a steady state:
		\begin{subequations}
			\label{eq:exp_model}
			\begin{align}
				y(t; \theta) &= \theta_1 + \sum_{n = 2}^{N}{e^{-\gamma_n t}} \\
				\gamma_n & \equiv \sum_{i = 2}^{n}{\theta_i}.
			\end{align}
		\end{subequations}
		Using $\theta_i$ as the parameters of the model, rather than using the decay rates $\gamma_n$ directly, guarantees that the decay rates are ordered, breaking the symmetry between them.
	\item A rational polynomial model:
		\begin{equation}
			\label{eq:thurber_model}
			y(t; \theta) = \frac{\theta_1 + \theta_2 t + \theta_3 t^2 + \theta_4 t^3}{1 + \theta_5 t + \theta_6 t^2 + \theta_7 t^3}.
		\end{equation}
	\item We used the IFFLP model of biological adaptation described in \cite{ma2009defining}.
	\item The FitzHugh-Nagumo model \cite{fitzhugh1961impulses, nagumo1962active, ramsay2007parameter} can be written as:
		\begin{subequations}
			\label{eq:fn_model}
			\begin{align}
				\dot{V} & = c\left(V - \frac{V^3}{3} + R\right) \\
				\dot{R} & = -\frac{1}{c}(V - a + bR).
			\end{align}
		\end{subequations}
		Initial conditions used were $(V_0, R_0) = (-1, 1)$.
	\item We implemented the Hodgkin-Huxley model described in \cite{hodgkin1952quantitative}.
	\item We used the Wnt oscillator model described in \cite{jensen2010wnt}.
	\item The Lorenz system \cite{lorenz1963deterministic} is given by:
		\begin{subequations}
			\label{eq:lorenz_model}
			\begin{align}
				\dot{x} & = \sigma(y - x) \\
				\dot{y} & = x(\rho - z) - y \\
				\dot{z} & = xy - \beta z.
			\end{align}
		\end{subequations}
		Initial conditions used were $(x_0, y_0, z_0) = (1, 1, 10)$.
		Additional parameters for rescaling $x$, $y$, and $z$ after solving the ODE were also included to illustrate that all parameters in a chaotic system need not exhibit an exponential sensitivity (see Fig.~\ref{fig:bigfig}(h)).
		In general, these parameters could account for differences in units between the model and the observations (if there were any).
	\item The Hindmarsh-Rose model \cite{hindmarsh1984model, wang1993genesis} can be written as:
		\begin{subequations}
			\label{eq:hr_model}
			\begin{align}
				\dot{x} & = y - a x^3 + b x^2 - z + I \\
				\dot{y} & = c - d x^2 - y \\
				\dot{z} & = \epsilon\left(x - \frac{1}{s}(z - z_R)\right).
			\end{align}
		\end{subequations}
		Initial conditions used were $(x_0, y_0, z_0) = (-0.216272..., 0.183969..., 0.066920...)$.
	\item The equations of motion for a damped, driven pendulum (derivable using Newton's \nth{2} law) are:
		\begin{subequations}
			\label{eq:ddp_model}
			\begin{align}
				\dot{\theta} & = \omega \\
				\dot{\omega} & = -\frac{\omega}{Q} - \sin(\theta) + A\cos(\phi) \\
				\dot{\phi}   & = \omega_D.
			\end{align}
		\end{subequations}
		Initial conditions used were $(\theta_0, \omega_0, \phi_0) = (-2, 0, 0)$.
\end{enumerate}
\section{Alternatives for obtaining a phase}
\label{sec:alt_phase}
In some cases, Eq.~\eqref{eq:AS_phase} cannot be used to obtain a monotonically increasing phase.
For example, some oscillatory behavior does not have a unique center of oscillation.
If that is the case, one approach is to decompose the signal using \emph{empirical mode decomposition} into a number of \emph{intrinsic mode functions}, for each of which a separate phase may then be defined \cite{huang1998empirical}.
However, because this method is empirical, the decomposition may not vary smoothly with the parameters of the model, leading to discontinuities in the cost function.

Even when the oscillatory behavior does have a single center of oscillation, in practice the Hilbert transform must be implemented numerically (especially for observational data).
This usually involves a fast Fourier transform, which can introduce unwanted effects in the phase due to the Gibbs phenomenon (see Fig.~\ref{fig:end_effects}).
The impact of end effects can be reduced by leaving the ends out of the cost function, or through windowing.

\begin{figure}[htbp]
	\includegraphics[width=1.0\columnwidth]{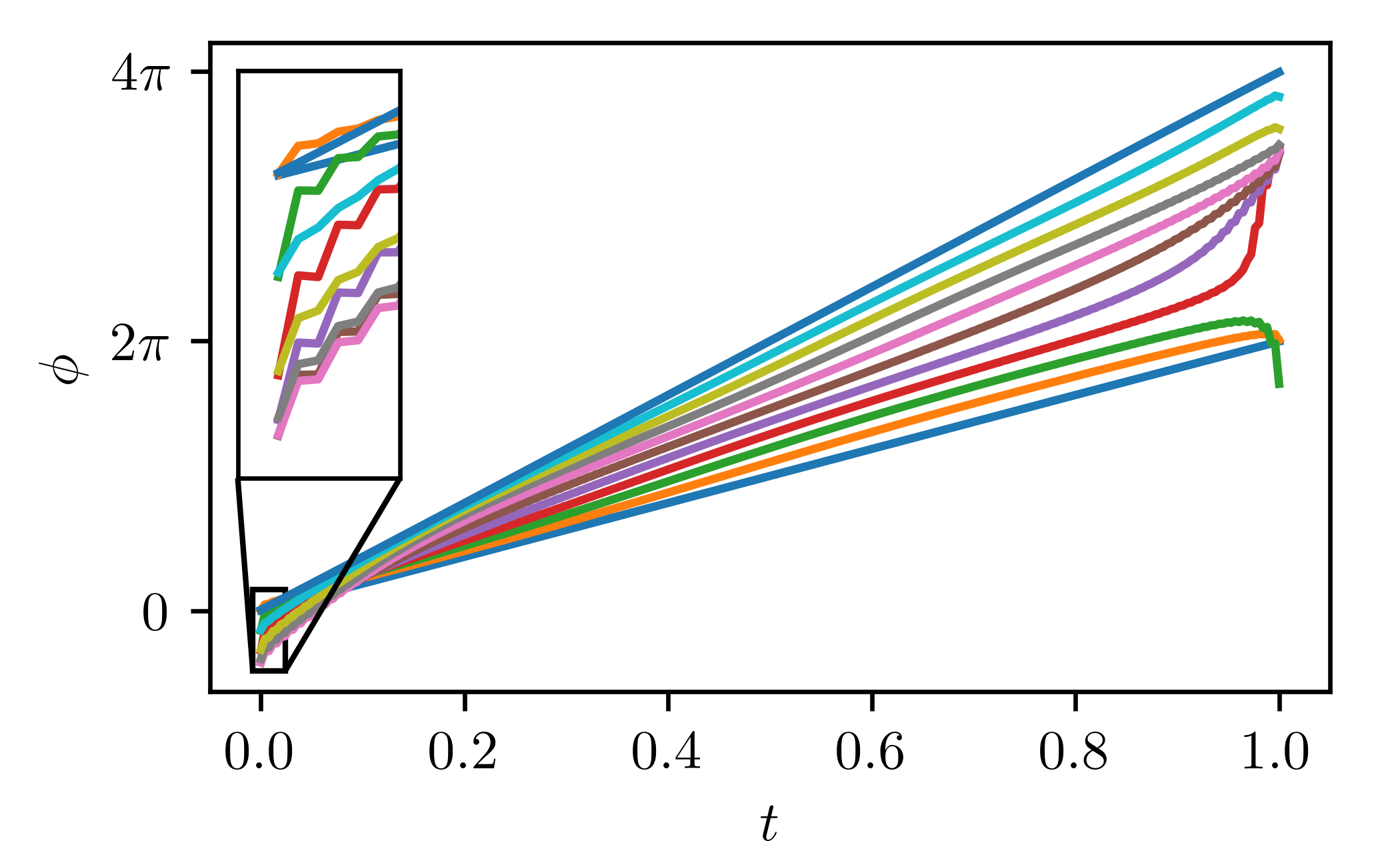}
	\caption{
		Phases obtained when implementing the Hilbert transform numerically on the model $y(t, \theta) = A\cos(\omega t)$, for $\omega$ ranging from $2\pi$ to $4\pi$.
		The effects of the Gibbs phenomenon can be seen near the ends for some values of $\omega$.
		}
	\label{fig:end_effects}
\end{figure}

More generally, any monotonically increasing function of time may be used for a phase, provided it has the appropriate frequency.
One proposal is to use
\begin{equation}
	\phi(t) = \omega t + \phi_0,
\end{equation}
and to estimate a value of $\omega$ from the oscillatory signal.
This may be done by fitting a line to the phase obtained from Eq.~\eqref{eq:AS_phase} or by using a Fourier transform to decompose the signal into frequency components and selecting one.

We also suggest the following method of obtaining a phase (found in \cite{pikovsky1997phase}) that does not require the use of the Hilbert transform.
It is sometimes the case that two signals, $y_1(t)$ and $y_2(t)$, can be selected from the dynamical variables $y(t)$ of a system and used to calculate a phase as follows:
\begin{equation}
	\phi(t) = \arctan\left(\frac{y_2(t)}{y_1(t)}\right).
\end{equation}
The only requirement is that the combined signal correspond to a \emph{proper rotation}, which has both a definite direction and unique center of rotation, so that the phase will be monotonically increasing \cite{yalccinkaya1997phase, kiss2001phase}.
For example, in some cases, a signal $y(t)$ and its time derivative $\dot{y}(t)$ may be used:
\begin{equation}
	\phi(t) = \arctan\left(\frac{y(t)}{\dot{y}(t)}\right).
\end{equation}
\section{Covariance matrices}
\label{sec:covariance}
We consider how uncertainty in experimental observations propagates to phases calculated using Eq.~\eqref{eq:AS_phase}.
First, we define more precisely the covariance matrix $\Sigma^{y(t)}$ for the observations with time as the independent variable.
Let $\xi_i$ denote random variables drawn from the normal distribution $\mathcal{N}(0, 1)$.
We assume the observations $y_i$ are random variables that are normally distributed about the predictions $y(t_i; \theta_0)$ of the model at the best fit, with standard deviation given by the uncertainties $\sigma_i$, and write
\begin{equation}
	\label{eq:data}
	y_i = y(t_i; \theta_0) + \sigma_i \xi_i.
\end{equation}
The deviations
\begin{equation}
	\label{eq:devs_t_app}
	\delta y_i \equiv y_i - y(t_i; \theta)
\end{equation}
vary with the predictions of the model, but at the best fit they are random variables with mean $0$ and standard deviation $\sigma_i$:
\begin{equation}
	\label{eq:devs_t_rand}
	\delta y_i(\theta_0) = y(t_i; \theta_0) + \sigma_i \xi_i - y(t_i; \theta_0) = \sigma_i \xi_i.
\end{equation}
The elements of the covariance matrix are defined as the expectation of the product of deviations at the best fit:
\begin{align}
	\Sigma^{y(t)}_{ij} & \equiv \left\langle \delta y_i \delta y_j \right\rangle \nonumber\\
	                   &      = \left\langle \sigma_i \xi_i \sigma_j \xi_j \right\rangle \nonumber\\
										 &      = \sigma_i \sigma_j \left\langle \xi_i \xi_j \right\rangle.
\end{align}
The matrix is diagonal if the deviations are independent (i.~e., if $\left\langle \xi_i \xi_j \right\rangle = \delta_{ij}$).
	\subsection{Covariance matrix for phase}
	\label{sec:cov_phase_vs_time}
	Next, we derive the covariance matrix for the phases.
The observations $y_i$ are assumed to have occurred at the phases $\phi(t_i; \theta_0)$ predicted by the model.
These phases will differ from the phases $\phi_i$ calculated using Eq.~\eqref{eq:AS_phase} due to the presence of noise in the observations.
We define the deviations of the phases as
\begin{equation}
	\label{eq:devs_phi_app}
	\delta \phi_i \equiv \phi_i - \phi(t_i; \theta).
\end{equation}
Note that, due to the presence of the Hilbert transform in Eq.~\eqref{eq:AS_phase}, the phase $\phi(t)$ has a functional dependence on the signal $y(t)$, i.e., $\phi(t) = \phi[y](t)$.
We use this functional dependence and Eq.~\eqref{eq:devs_t_app} to relate $\delta \phi_i$ to $\delta y_i$:
\begin{align}
	\phi_i & = \phi_i[y] \nonumber\\
				 & = \phi_i[y(t; \theta) + \delta y] \nonumber\\
				 & \approx \phi_i[y(t; \theta)] + \sum_j{\frac{\partial \phi_i[y(t; \theta)]}{\partial y_j}\delta y_j} \nonumber\\
				 & = \phi(t_i; \theta) + \sum_j{\frac{\partial \phi_i}{\partial y_j}\delta y_j}, \\
	\label{eq:devs_phi}
	\delta \phi_i & \approx \sum_j{\frac{\partial \phi_i}{\partial y_j}\delta y_j}.
\end{align}
In the fourth line we have simplified the notation for clarity, and we have kept only the first order terms.
This approximation is valid near the best fit where $\delta y_i$ is small.
At the best fit, we have
\begin{equation}
	\label{eq:devs_phi_rand}
	\delta \phi_i(\theta_0) = \sum_j{\frac{\partial \phi_i}{\partial y_j}\sigma_j \xi_j},
\end{equation}
which shows that $\delta \phi_i(\theta_0)$ are random variables with mean $0$.

Before proceeding, the derivative $\partial \phi_i/\partial y_j$ merits some attention.  First, we note that it may be evaluated using either $y(t_i; \theta)$ or $y_i$ to first order in $\delta y_i$:
\begin{align}
	\frac{\partial \phi_i[y(t_j; \theta)]}{\partial y_j}\delta y_j
		& = \frac{\partial \phi_i[y_j - \delta y_j]}{\partial y_j}\delta y_j \nonumber\\
		& = \frac{\partial \phi_i[y_j]}{\partial y_j}\delta y_j + O(\delta y^2)
\end{align}

Second, using Eq.~\eqref{eq:AS_phase}, we can derive an explicit expression for $\partial \phi_i/\partial y_j$:
\begin{align}
	\frac{\partial \phi_i}{\partial y_j} & = \frac{\partial}{\partial y_j}\left[\tan^{-1}\left(\frac{H_i[y]}{y_i}\right)\right] \nonumber\\
	\label{eq:dphdy_middle}
		& = \frac{1}{1 + \left(H_i[y]/y_i\right)^2}\left(\frac{1}{y_i}\frac{\partial H_i[y]}{\partial y_j} - \frac{H_i[y]}{y_i^2}\frac{\partial y_i}{\partial y_j}\right).
\end{align}
($H_i[y]$ is understood to mean the $i$th component of the Hilbert transform of $y$.)
To evaluate the derivative $\partial H_i[y]/\partial y_j$, we use the definition of the derivative and the linearity of the Hilbert transform:
\begin{align}
	\frac{\partial H_i[y]}{\partial y_j} & = \lim_{h \to 0}\frac{H_i[y + h\delta_j] - H_i[y]}{h} \nonumber\\
																			 & = \lim_{h \to 0}\frac{H_i[y] + h H_i[\delta_j] - H_i[y]}{h} \nonumber\\
																			 & = H_i[\delta_j].
\end{align}
(We are using $\delta_j$ to denote the vector formed by taking the $j$th column of the Kronecker delta $\delta_{ij}$ when considered as a matrix.)
Plugging this into Eq.~\eqref{eq:dphdy_middle} gives
\begin{equation}
	\label{eq:dphdy}
	\frac{\partial \phi_i}{\partial y_j} = \frac{y_i H_i[\delta_j] - H_i[y] \delta_{ij}}{y_i^2 + H_i[y]^2}.
\end{equation}

Third, the matrix $\partial \phi/\partial y$ defined by Eq.~\eqref{eq:dphdy} is singular (i.~e., it has at least one zero eigenvalue).
As we now show, this is because changes in the amplitude of an oscillation do not affect the phase.
\begin{thm*}
	The matrix $\partial \phi/\partial y$, whose $ij$th element is
	\begin{equation*}
		\frac{\partial \phi_i}{\partial y_j} = \frac{y_i H_i[\delta_j] - H_i[y] \delta_{ij}}{y_i^2 + H_i[y]^2}
	\end{equation*}
	has at least one zero eigenvalue, corresponding to the eigenvector $\delta y^* = y$.
\end{thm*}
\begin{proof}
	\begin{align*}
		\sum_j{\frac{\partial \phi_i}{\partial y_j} \delta y_j^*}
			& = \sum_j{\frac{y_i H_i[\delta_j] - H_i[y] \delta_{ij}}{y_i^2 + H_i[y]^2} y_j} \\
			& = \frac{y_i H_i\left[\textstyle{\sum_j}{\delta_j y_j}\right] - H_i[y] \textstyle{\sum_j}{\delta_{ij} y_j}}{y_i^2 + H_i[y]^2} \\
			& = \frac{y_i H_i[y] - H_i[y] y_i}{y_i^2 + H_i[y]^2} \\
			& = 0. \qedhere
	\end{align*}
\end{proof}
Any change in amplitude at constant phase is a multiple of $y$ and thus also lies in the null space of $\partial \phi/\partial y$.

Returning to Eq.~\eqref{eq:devs_phi}, we derive an expression for the covariance matrix $\Sigma^{\phi(t)}$ for the phases:
\begin{align}
	\Sigma^{\phi(t)}_{ij} & \equiv \left\langle \delta \phi_i \delta \phi_j \right\rangle \nonumber\\
												& = \left\langle \sum_k{\frac{\partial \phi_i}{\partial y_k} \delta y_k} \sum_l{\frac{\partial \phi_j}{\partial y_l} \delta y_l} \right\rangle \nonumber\\
												& = \sum_{k, l}{\frac{\partial \phi_i}{\partial y_k} \left\langle \delta y_k \delta y_l \right\rangle \frac{\partial \phi_j}{\partial y_l}} \nonumber\\
												& = \sum_{k, l}{\frac{\partial \phi_i}{\partial y_k} \Sigma^{y(t)}_{kl} \frac{\partial \phi_j}{\partial y_l}} \\
	\Sigma^{\phi(t)} & = \frac{\partial \phi}{\partial y} \Sigma^{y(t)} \frac{\partial \phi}{\partial y}^T. \label{eq:cov_phi_t}
\end{align}
This shows how uncertainties $\sigma_{y_i}^2 = \Sigma^{y(t)}_{ii}$ in the observations are propagated to uncertainties $\sigma_{\phi_i}^2 = \Sigma^{\phi(t)}_{ii}$ in the phases of the observations (see Fig.~\ref{fig:sin_stats}).

\begin{figure}[htbp]
	\includegraphics[width=1.0\columnwidth]{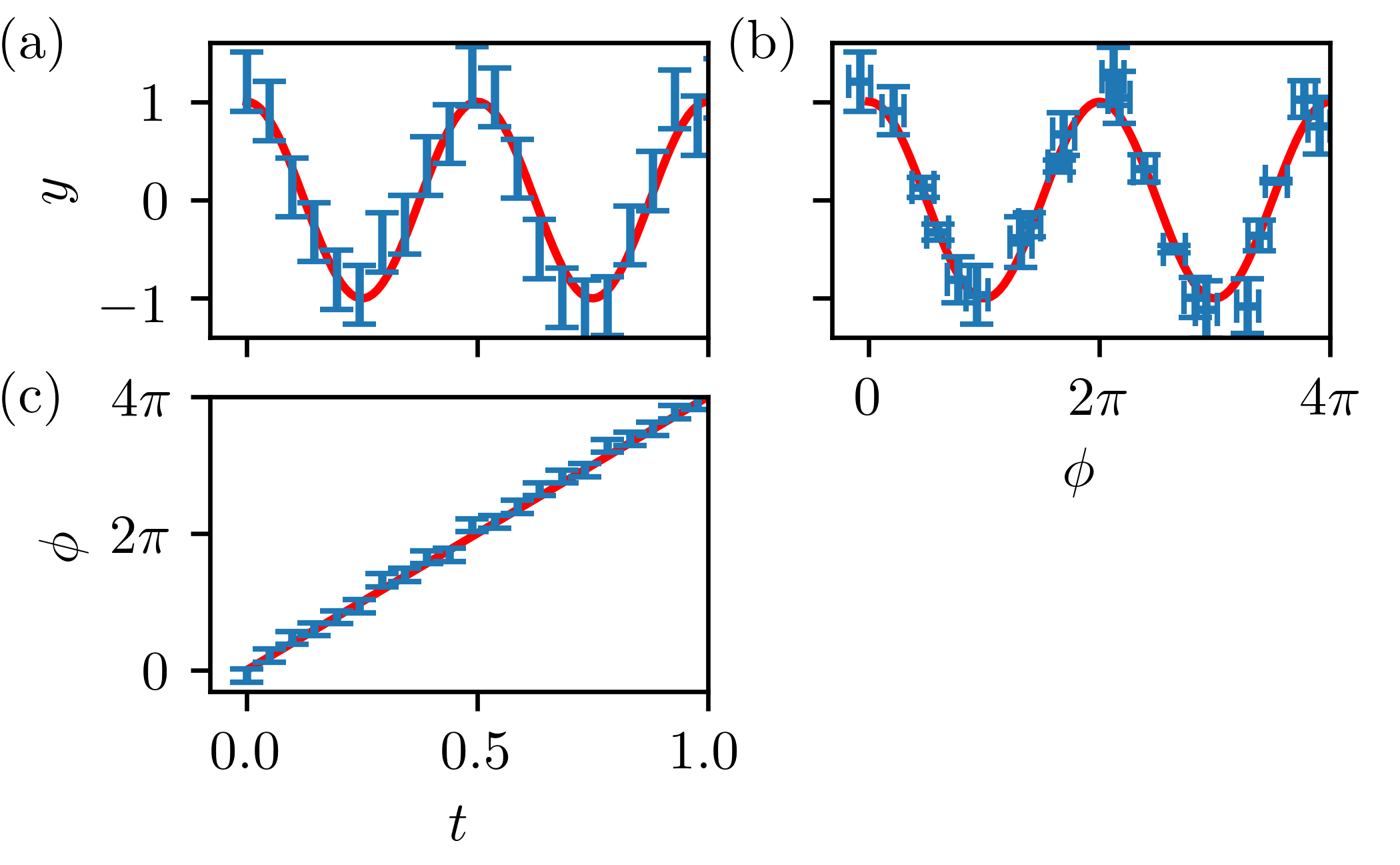}
	\caption{
		\textbf{Propagation of uncertainty.}
		(a) Data (blue) simulated from the model $y(t) = A\cos(\omega t)$ (red) by adding uniform Gaussian noise.
		Error bars indicate uncertainty.
		(b) Data (blue) plotted as a function of phase compared with $y(\phi) = A\cos(\phi)$ (red).
		Error bars indicate the uncertainties obtained using Eqs.~\eqref{eq:cov_phi_t} and \eqref{eq:cov_y_phi}.
		(c) Phase (blue), obtained for each data point using Eq.~\eqref{eq:AS_phase}, compared with $\phi(t) = \omega t$ (red).
		Error bars indicate the uncertainties obtained using Eq.~\eqref{eq:cov_phi_t}.
		}
	\label{fig:sin_stats}
\end{figure}
	\subsection{Covariance matrix for observations as a function of phase}
	\label{sec:cov_y_vs_phase}
	Finally, we derive the covariance matrix $\Sigma^{y(\phi)}$ for the observations with phase as the independent variable.
We define the deviations of the observations from the predictions at constant phase as
\begin{equation}
	\label{eq:devs_ytilde_app}
	\delta \tilde{y}_i \equiv y_i - \tilde{y}(\phi_i; \theta).
\end{equation}
We can relate these to $\delta y_i$ and $\delta \phi_i$ using Eqs.~\eqref{eq:devs_phi_app} and \eqref{eq:time_to_phase}:
\begin{align}
	\delta \tilde{y}_i & = y_i - \tilde{y}(\phi_i; \theta) \nonumber\\
		& = y_i - \tilde{y}\bm{(}\phi(t_i; \theta) + \delta \phi_i; \theta\bm{)} \nonumber\\
		& \approx y_i - \tilde{y}\bm{(}\phi(t_i; \theta), \theta\bm{)} - \frac{\partial \tilde{y}\bm{(}\phi(t_i; \theta); \theta\bm{)}}{\partial \phi}\delta \phi_i \nonumber\\
		& = y_i - y(t_i; \theta) - \left(\frac{\partial \tilde{y}}{\partial \phi}\right)_i\delta \phi_i \nonumber\\
		& = \delta y_i - \left(\frac{\partial \tilde{y}}{\partial \phi}\right)_i\delta \phi_i. \label{eq:devs_ytilde}
\end{align}
In light of Eqs.~\eqref{eq:devs_t_rand} and \eqref{eq:devs_phi_rand}, $\delta \tilde{y}_i$ also has mean $0$ at the best fit.
Note that, similar to $\partial \phi_i/\partial y_j$, $\partial \tilde{y}/\partial \phi$ may be evaluated using either $\phi_i$ or $\phi(t_i; \theta)$ to first order in $\delta \phi_i$:
\begin{align}
	\frac{\partial \tilde{y}\bm{(}\phi(t_i; \theta); \theta\bm{)}}{\partial \phi}\delta \phi_i
		& = \frac{\partial \tilde{y}(\phi_i - \delta \phi_i; \theta)}{\partial \phi}\delta \phi_i \nonumber\\
		& = \frac{\partial \tilde{y}(\phi_i; \theta)}{\partial \phi}\delta \phi_i + \mathcal{O}(\delta \phi^2)
\end{align}

We can take Eq.~\eqref{eq:devs_ytilde} a step further using Eq.~\eqref{eq:devs_phi}:
\begin{align}
	\delta \tilde{y}_i & = \delta y_i - \left(\frac{\partial \tilde{y}}{\partial \phi}\right)_i\delta \phi_i \nonumber\\
												 & = \delta y_i - \sum_j{\left(\frac{\partial \tilde{y}}{\partial \phi}\right)_i\frac{\partial \phi_i}{\partial y_j}\delta y_j} \nonumber\\
												 & = \sum_j{\left[\delta_{ij} - \left(\frac{\partial \tilde{y}}{\partial \phi}\right)_i\frac{\partial \phi_i}{\partial y_j}\right]\delta y_j} \nonumber\\
												 & \equiv \sum_j{D_{ij}\delta y_j}.
\end{align}
Taking the expectation of pairs of deviations $\delta \tilde{y}_i$, we obtain
\begin{align}
	\Sigma^{y(\phi)}_{ij} & \equiv \left\langle \delta \tilde{y}_i \delta \tilde{y}_j \right\rangle \nonumber\\
												& = \left\langle \sum_k{D_{ik}\delta y_k} \sum_l{D_{jl}\delta y_l} \right\rangle \nonumber\\
												& = \sum_{k,l}{D_{ik} \left\langle \delta y_k \delta y_l \right\rangle D_{jl}} \nonumber\\
												& = \sum_{k,l}{D_{ik} \Sigma^{y(t)}_{kl} D_{jl}} \\
	\Sigma^{y(\phi)} & = D \Sigma^{y(t)} D^T. \label{eq:cov_y_phi}
\end{align}
This gives us a way to compute the uncertainties $\sigma_{\tilde{y}_i}^2 = \Sigma^{y(\phi)}_{ii}$ of the observations when taking phase as the independent variable instead of time (see Fig.~\ref{fig:sin_stats}).
\section{Parameter sensitivities}
\label{sec:derivatives}
Here we derive the first- and second-order parameter sensitivities of $\tilde{y}$ and $\phi$ that are used in calculating the FIM and winding frequencies for the analytic signal-based metric of Sec.~\ref{sec:AS}.
We begin with Eq.~\eqref{eq:amp_and_phase},
\begin{equation}
	\label{eq:amp_and_phase2}
	y(t; \theta) = A(t; \theta)\cos\bm{(}\phi(t; \theta)\bm{)},
\end{equation}
and differentiate it with respect to $\theta_\mu$:
\begin{equation}
	\frac{\partial y}{\partial \theta_\mu} = \frac{\partial A}{\partial \theta_\mu}\cos(\phi) - A\sin(\phi)\frac{\partial \phi}{\partial \theta_\mu}.
\end{equation}
Comparing with Eq.~\eqref{eq:decoupling2}, we now see that we have explicit expressions for $\partial \tilde{y}/\partial \theta_\mu$ and $\partial \tilde{y}/\partial \phi$ in terms of $A$, $\phi$, and $\partial A/\partial \theta_\mu$:
\begin{align}
	\left.\frac{\partial \tilde{y}}{\partial \theta_\mu}\right|_\phi &= \frac{\partial A}{\partial \theta_\mu}\cos(\phi) &
	\left.\frac{\partial \tilde{y}}{\partial \phi}\right|_\theta &= -A\sin(\phi) \nonumber\\
		&= \frac{y}{A}\frac{\partial A}{\partial \theta_\mu} &
		&= -H[y].
\end{align}
In the second line we have used the trigonometric relationships $\cos(\phi) = y/A$ and $\sin(\phi) = H[y]/A$ which are easily derived from Eqs.~\eqref{eq:AS_amp} and \eqref{eq:AS_phase}.
The second derivative of Eq.~\eqref{eq:amp_and_phase2} is
\begin{multline}
	\frac{\partial^2 y}{\partial \theta_\mu\partial \theta_\nu} = \frac{y}{A}\frac{\partial^2 A}{\partial \theta_\mu\partial \theta_\nu} - \frac{H[y]}{A}\left(\frac{\partial A}{\partial \theta_\mu}\frac{\partial \phi}{\partial \theta_\nu} + \frac{\partial \phi}{\partial \theta_\mu}\frac{\partial A}{\partial \theta_\nu}\right) \\
		- y\frac{\partial \phi}{\partial \theta_\mu}\frac{\partial \phi}{\partial \theta_\nu} - H[y]\frac{\partial^2 \phi}{\partial \theta_\mu\partial \theta_\nu}.
\end{multline}
Because the new analytic signal-based metric involves $\tilde{y}$ and $\phi$, we use only the first term (which is $\partial^2 \tilde{y}/\partial \theta_\mu\partial \theta_\nu$) and the last term in this expression when calculating the geodesic curvature.

Expressions for the sensitivities of $A$ and $\phi$ are obtained by differentiating Eqs.~\eqref{eq:AS_amp} and \eqref{eq:AS_phase}:
\begin{align}
	A &= \sqrt{y^2 + H^2[y]} & \phi &= \tan^{-1}\left(\frac{H[y]}{y}\right)
\end{align}
\begin{align}
	\frac{\partial A}{\partial \theta_\mu} &= \frac{1}{A}\left(y\frac{\partial y}{\partial \theta_\mu} + H[y]H\left[\frac{\partial y}{\partial \theta_\mu}\right]\right) \\
	\frac{\partial \phi}{\partial \theta_\mu} &= \frac{1}{A^2}\left(yH\left[\frac{\partial y}{\partial \theta_\mu}\right] - H[y]\frac{\partial y}{\partial \theta_\mu}\right)
\end{align}
\begin{multline}
	\frac{\partial^2 A}{\partial \theta_\mu\partial \theta_\nu} = A\frac{\partial \phi}{\partial \theta_\mu}\frac{\partial \phi}{\partial \theta_\nu} + \frac{1}{A}\left(y\frac{\partial^2 y}{\partial \theta_\mu\partial \theta_\nu}\right. \\
		\left.+ H[y]H\left[\frac{\partial^2 y}{\partial \theta_\mu\partial \theta_\nu}\right]\right)
\end{multline}
\begin{multline}
	\frac{\partial^2 \phi}{\partial \theta_\mu\partial \theta_\nu} = -\frac{1}{A}\frac{\partial A}{\partial \theta_\mu}\frac{\partial \phi}{\partial \theta_\nu} - \frac{1}{A}\frac{\partial \phi}{\partial \theta_\mu}\frac{\partial A}{\partial \theta_\nu} \\
		+ \frac{1}{A^2}\left(yH\left[\frac{\partial^2 y}{\partial \theta_\mu\partial \theta_\nu}\right] - H[y]\frac{\partial^2 y}{\partial \theta_\mu\partial \theta_\nu}\right)
\end{multline}
\section{Regularity of cost surfaces and manifolds}
\label{sec:convergence}
In Fig.~\ref{fig:resultsfig}, a sufficiently large number of time points was included in the cost and manifold calculations to demonstrate the results of using the new metrics in the limit of infinite time.
In practice, only a finite number of time points can be included.
Here we demonstrate the convergence of the FitzHugh-Nagumo manifold and the Lorenz cost as a function of the number of sampled time points.
In addition, we discuss the gradient of the FitzHugh-Nagumo cost, shown in Fig.~\ref{fig:resultsfig}(a), as it relates to the regularity of the new surface.

Figure \ref{fig:cost_grad} shows a plot of the magnitude of the gradient of the cost cross section shown in Fig.~\ref{fig:resultsfig}(a).
The significance of the gradient of the cost is that every local minimum of the cost will be a zero of the gradient.
If there are multiple local minima still present in the new cost, then the gradient will have multiple zeros.
We plot the magnitude of the gradient so that zeros can be found easily.
It is clear from Fig.~\ref{fig:cost_grad} that there is only one zero, so the new cost does, in fact, have a single minimum.

\begin{figure}[htbp]
	\includegraphics[width=1.0\columnwidth]{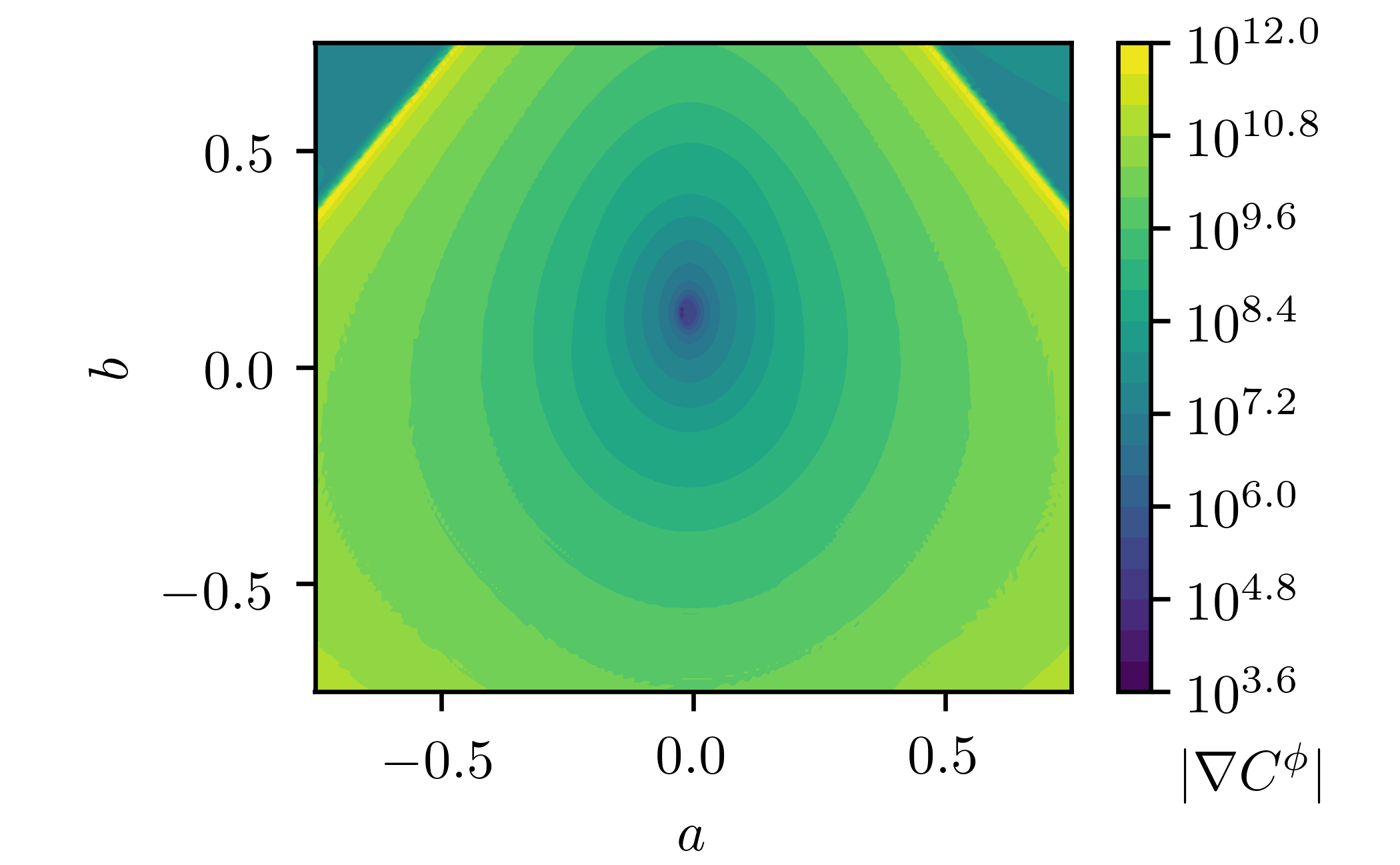}
	\caption{
		\textbf{Magnitude of the gradient of the FitzHugh-Nagumo cost} $|\nabla C^\phi(\theta)|$.
		The magnitude of the gradient has only one minimum, indicating that the cost cross section shown in Fig.~\ref{fig:resultsfig}(a) has a single minimum.
		The minimum of $|\nabla C^\phi(\theta)|$ shown is not quite zero because the actual minimum of Fig.~\ref{fig:resultsfig}(a) is between the grid points where $|\nabla C^\phi(\theta)|$ has been calculated.
		Note that in the upper corners of the plot, there is a phase transition to nonoscillatory behavior, where the methods of Sec.~\ref{sec:AS} cannot be applied effectively.
		The sharp apparent dropoff is due to such choices as having our algorithms return zeros rather than throw errors for these regions.
		}
	\label{fig:cost_grad}
\end{figure}
				
Figure \ref{fig:fn_conv} shows two projections of the FitzHugh-Nagumo manifold (signal predictions at constant phase): one calculated using about 24 time points per cycle in the original time series and the other using twice the time sampling of the first.
The manifold itself exhibits oscillations in both cases.
These oscillations are an artifact of the finite time sampling of the oscillatory signal predicted by the model.
As parameters that control frequency are varied, the peak of each cycle shifts between adjacent time points and the local amplitude appears to oscillate (see Fig.~\ref{fig:fn_oscillations}).
Hence the predicted signal values at a given (constant) phase also oscillate, resulting in the manifold oscillations observed.

\begin{figure}[htbp]
	\includegraphics[width=1.0\columnwidth]{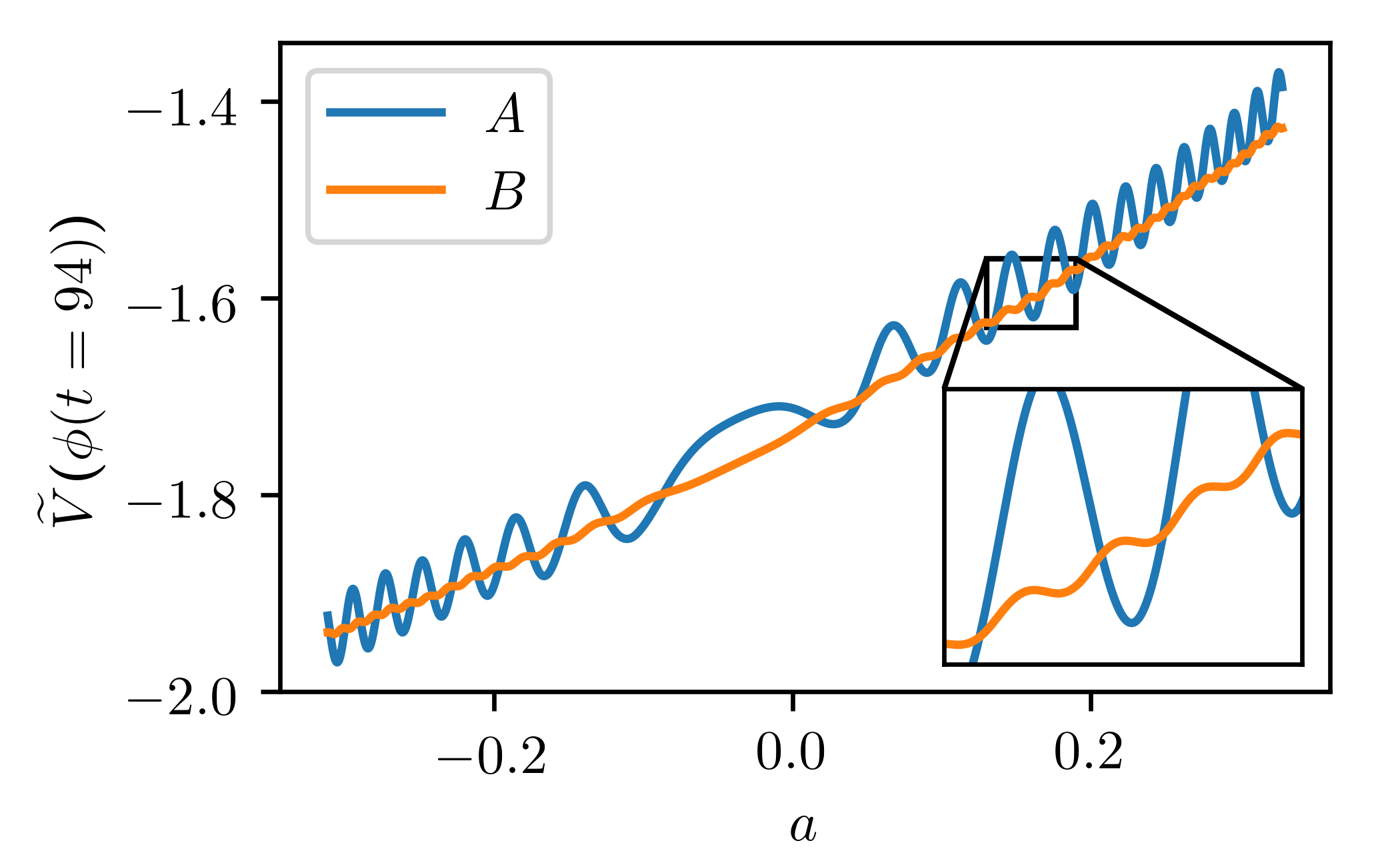}
	\caption{
		\textbf{FitzHugh-Nagumo manifold projection.}
		$A$ was calculated using about 24 time points per cycle in the original time series; $B$ was calculated using twice the time sampling of $A$.
		}
	\label{fig:fn_conv}
\end{figure}

\begin{figure}[htbp]
	\includegraphics[width=1.0\columnwidth]{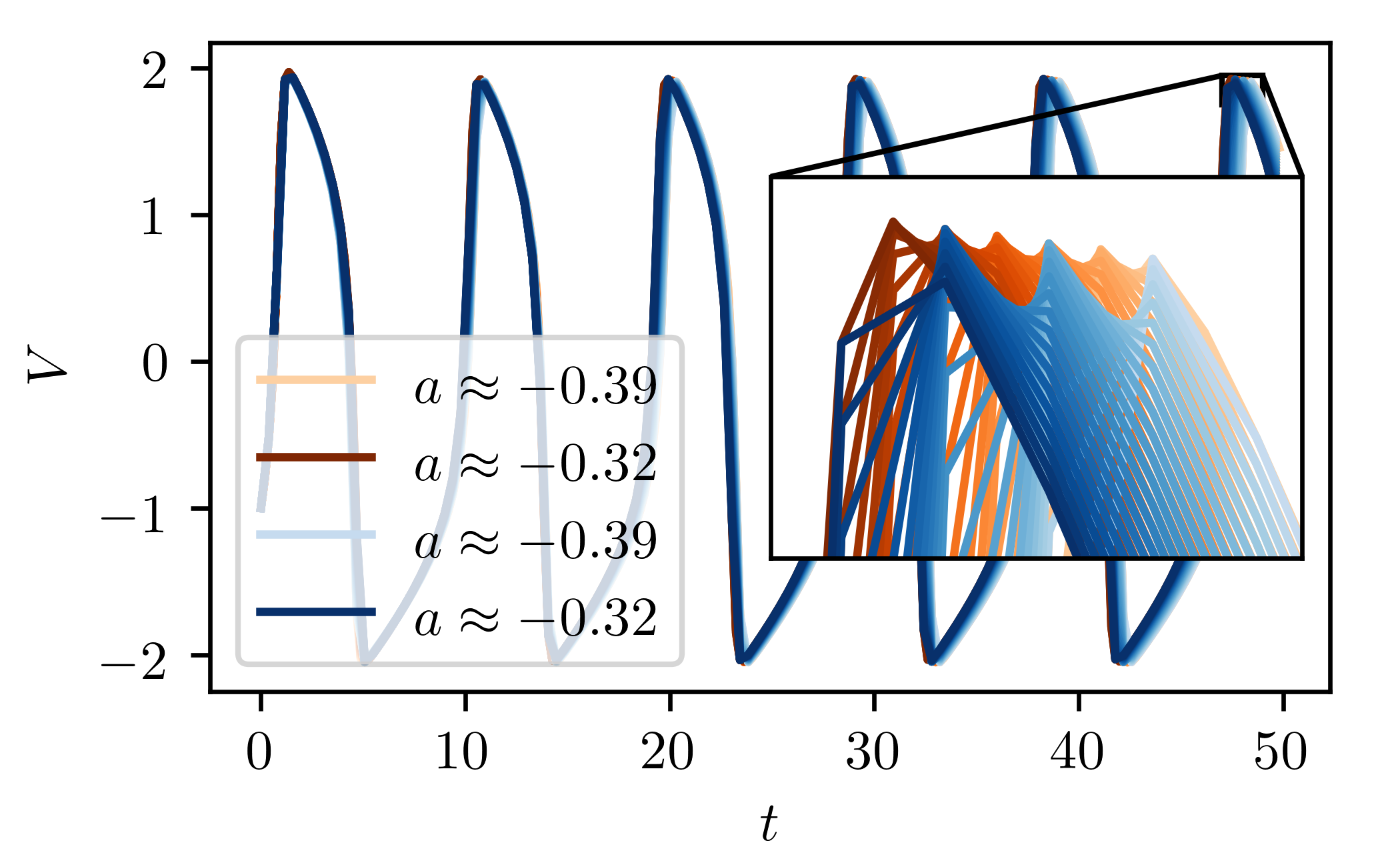}
	\caption{
		\textbf{FitzHugh-Nagumo amplitude oscillations.}
		Colors are the same as in Fig.~\ref{fig:fn_conv}, with dark/light indicating the value of the parameter $a$.
		As the peak moves between sampled time points, the amplitude appears to oscillate.
		}
	\label{fig:fn_oscillations}
\end{figure}

As demonstrated in Fig.~\ref{fig:fn_conv}, doubling the sampling of time points doubles the frequency of these manifold oscillations, but their amplitude decreases by a factor of $\sim 10$.
Hence, in the limit of infinite sampling they disappear.
In practice they will be negligible as long as enough time points per cycle are sampled for the amplitude of the oscillations to be small compared to the amplitude of the signal itself (and to changes effected by the parameters).

The attractors of chaotic systems have fractal structure that is realized only in the limit of infinite sampling time $T$.
Accordingly, as more time points are included, the kernel density estimate Eq.~\eqref{eq:kde} will approach the true distribution $f(y, \theta)$ asymptotically.
Figure \ref{fig:lorenz_conv} illustrates the convergence of a cross section of the cost Eq.~\eqref{eq:cost_KDE} for the Lorenz system as the total sampling time $T$ is varied.

\begin{figure}[htbp]
	\includegraphics[width=1.0\columnwidth]{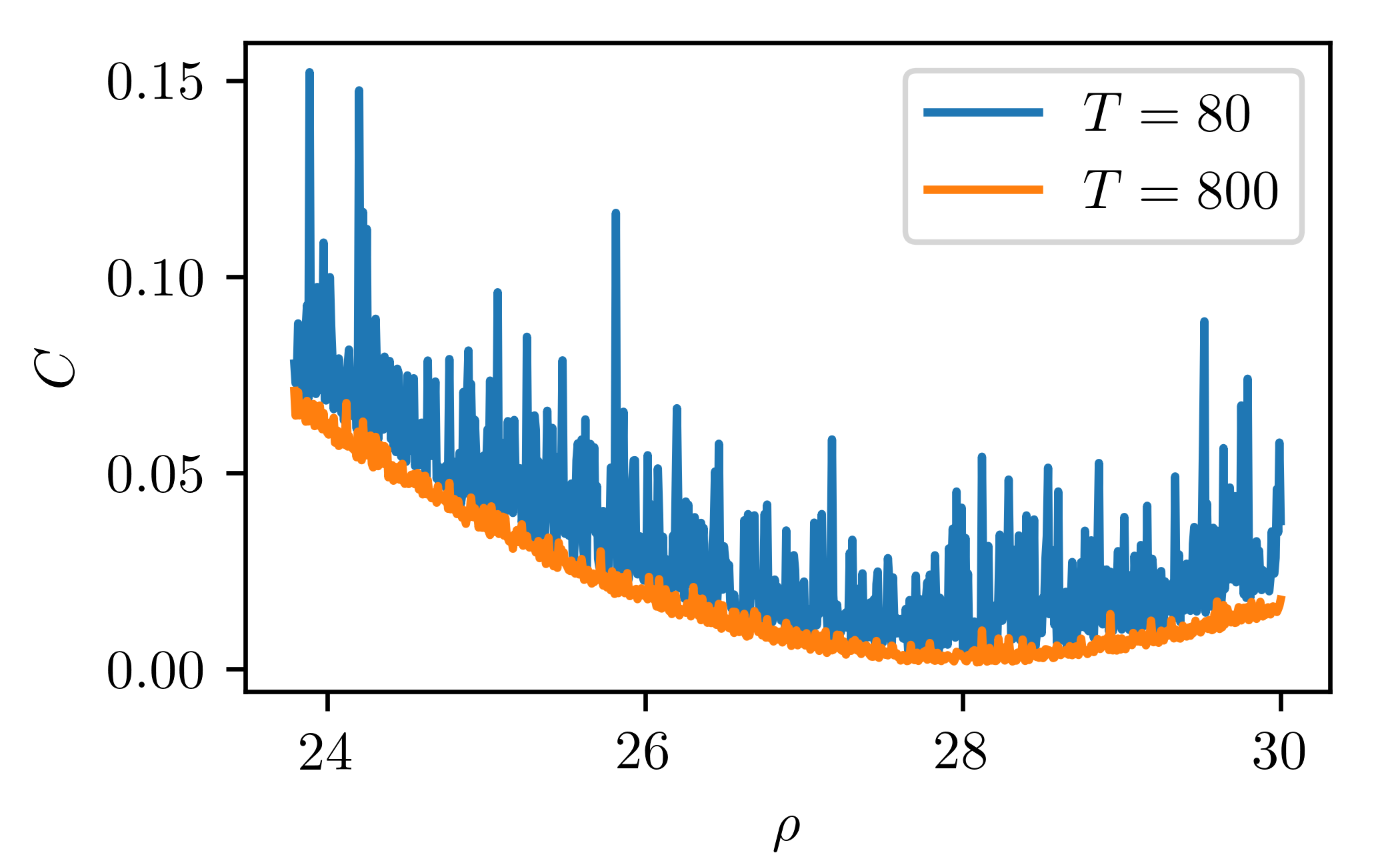}
	\caption{
		\textbf{Lorenz cost.}
		As the number of sampled time points grows, the noise in the cost dies away.
		When fit to a parabola, the MAE between the parabola and the cost cross section shown is 0.0088 for $T = 80$ and 0.0013 for $T = 800$ (about a sevenfold reduction in noise).
		}
	\label{fig:lorenz_conv}
\end{figure}

\end{document}